\documentclass[a4paper,UKenglish,cleveref,autoref]{lipics-v2019}

\usepackage{amsmath}
\usepackage{hyperref}\renewcommand\url[1]{\href{#1}{\tt #1}}
\usepackage{amssymb} % \blacksquare, \mathfrak
\usepackage{t1enc}
\usepackage{stmaryrd} % \llbracket
\usepackage{graphicx}
\usepackage{diagbox}
\usepackage{multirow}
\usepackage{tikz}
\usetikzlibrary{arrows,shapes,decorations.pathmorphing,backgrounds,positioning,fit,calc}
\usepackage{bussproofs}
\usepackage{xspace} % xspace
\usepackage{enumerate}
\usepackage{mathtools}

\newcommand\tool[1]{{\scshape #1}\xspace}

\newcommand\moins{\setminus}

\newcommand\furl[1]{\footnote{\url{http://#1}}}
\newcommand\hide[1]{}

\newcommand\dom{\mr{dom}}

\newcommand\FV{\mr{FV}}
\newcommand\SN{\mr{SN}}

\renewcommand\prod{\mr{prod}}

\renewcommand\a{\rightarrow}
\newcommand\ra{\longrightarrow}
\newcommand\A{\Rightarrow}

\newcommand\la{\leftarrow}

\newcommand\ad{\downarrow}

\renewcommand\to{\mapsto}
\newcommand\nf[1]{{{#1}\!\!\ad}}

\newcommand\ab{\a_\b}

\newcommand\ar{\a_\cR}

\newcommand\I[1]{\llbracket{#1}\rrbracket}% stmaryrd

\newcommand\all{\forall}

\renewcommand\th{\vdash}

\newcommand\sle{\subseteq}

\newcommand\tgt{\rhd}

\newcommand\al{\alpha}
\renewcommand\b{\beta}
\newcommand\g{\gamma}
\newcommand\G{\Gamma}

\newcommand\D{\Delta}

\renewcommand\t{\theta}
\newcommand\T{\Theta}

\renewcommand\l{\lambda}

\newcommand\s{\sigma}
\renewcommand\S{\Sigma}

\newcommand\w{\omega}

\newcommand\mc{\mathcal}

\newcommand\mr{\mathrm}
\newcommand\mb{\mathbb}

\newcommand\bB{\mb{B}}
\newcommand\bC{\mb{C}}

\newcommand\bF{\mb{F}}

\newcommand\bI{\mb{I}}

\newcommand\bK{\mb{K}}
\newcommand\bL{\mb{L}}

\newcommand\bN{\mb{N}}

\newcommand\bP{\mb{P}}

\newcommand\bS{\mb{S}}
\newcommand\bT{\mb{T}}

\newcommand\bV{\mb{V}}

\newcommand\cD{\mc{D}}

\newcommand\cF{\mc{F}}
\newcommand\cG{\mc{G}}

\newcommand\cI{\mc{I}}

\newcommand\cP{\mc{P}}

\newcommand\cR{\mc{R}}

\newcommand\vl{{\vec{l}}}
\newcommand\vm{{\vec{m}}}

\newcommand\vt{{\vec{t}}}
\newcommand\vu{{\vec{u}}}
\newcommand\vv{{\vec{v}}}

\newcommand\vx{{\vec{x}}}
\newcommand\vy{{\vec{y}}}

\newcommand\vA{{\vec{A}}}

\newcommand\vT{{\vec{T}}}
\newcommand\vU{{\vec{U}}}

\newenvironment{rul}
  {$\begin{array}{rcl}}
  {\end{array}$}

\newenvironment{rew}[1][~~\a~~]
  {$\begin{array}{r@{#1}l}}
  {\end{array}$}
\newenvironment{rewc}[1][~~\a~~]
  {\begin{center}\begin{rew}[#1]}
  {\end{rew}\end{center}}

\renewcommand\prod[3]{\all #1:#2,#3}
\newcommand\abs[3]{\l #1:#2,#3}
\newcommand\type{\text{\footnotesize$\mr{TYPE}$}}
\newcommand\kind{\text{\footnotesize$\mr{KIND}$}}

\renewcommand\ra\a
\renewcommand\arg{\mr{arg}}

\newcommand\red[1]{\a\!\!(#1)}
\newcommand\call{\tilde{>}}
\newcommand\<\leq
\DeclareMathOperator\arit{ar}
\newcommand\product[3]{\Pi #1\hspace{-0.15em}\in \hspace{-0.15em} #2.\,#3}

\definecolor{green}{RGB}{0,130,0}
\definecolor{grispale}{RGB}{240,240,240}

\usepackage{listings}
\lstdefinelanguage{Dedukti}
{
  alsoletter={=->:\#\*},
  keywords={symbol,rule,infix,TYPE,set},
  delim=[s][\color{brown}]{\[}{\]},
  comment=[n]{(;}{;)},
  string=[b]{"},
  commentstyle=\color{red},
  showstringspaces=false
}
\definecolor{darkorange}{RGB}{180,100,20}
\newcommand{\alp}{{\color{darkorange}\rightarrow}}

\title{Dependency Pairs Termination in Dependent Type Theory Modulo Rewriting}

\author{Fr\'ed\'eric Blanqui$^{1,2}$ \and Guillaume Genestier$^{2,3}$ \and Olivier Hermant$^3$}{$^1$ INRIA \and $^2$ LSV, ENS Paris-Saclay, CNRS, Universit\'e Paris-Saclay \and $^3$ MINES ParisTech, PSL University}{}{}{}

\authorrunning{F. Blanqui, G. Genestier and O. Hermant}
\Copyright{Fr\'ed\'eric Blanqui, Guillaume Genestier and Olivier Hermant}



\begin{document}

\maketitle

\begin{abstract}
  Dependency pairs are a key concept at the core of modern automated
  termination provers for first-order term rewriting systems. In this
  paper, we introduce an extension of this technique for a large class
  of dependently-typed higher-order rewriting systems. This
  extends previous results by Wahlstedt on the one hand and the first
  author on the other hand to strong normalization and non-orthogonal
  rewriting systems. This new criterion is implemented in the
  type-checker \tool{Dedukti}.
\end{abstract}

\section{Introduction}

Termination, that is, the absence of infinite computations, is an
important problem in software verification, as well as in logic.
In logic, it is often used to prove cut elimination and
consistency. In automated theorem provers and proof assistants, it is
often used (together with confluence) to check decidability of
equational theories and type-checking algorithms.

This paper introduces a new termination criterion for a large class of
programs whose operational semantics can be described by higher-order
rewriting rules \cite{terese03book} typable in the $\l\Pi$-calculus
modulo rewriting ($\l\Pi/\cR$ for short). $\l\Pi/\cR$ is a system of
dependent types where types are identified modulo the $\b$-reduction
of $\l$-calculus and a set $\cR$ of rewriting rules given by the user to
define not only functions but also types. It extends
Barendregt's Pure Type System (PTS) $\l P$ \cite{barendregt92chapter},
the logical framework LF \cite{harper93jacm} and Martin-L\"of's type
theory\hide{\cite{martinlof84book}}. It can encode any
functional PTS like System F or the Calculus of Constructions
\cite{cousineau07tlca}.\hide{assaf15phd} \hide{This makes $\l\Pi/\cR$ a good
  candidate for a logical framework and a platform for proof
  translation and interoperability \cite{assaf16draft}.}

Dependent types, introduced by de Bruijn in \tool{Automath}\hide{\cite{debruijn68sad}}, subsume
generalized algebraic data types (GADT)\hide{\cite{xi03popl}} used in some functional programming languages.
They are at the core of many proof
assistants and programming languages: \tool{Coq}, \tool{Twelf},
\tool{Agda}, \tool{Lean}, \tool{Idris}, \ldots

Our criterion has been implemented in \href{https://deducteam.github.io/}{\tool{Dedukti}}\hide{\cite{dedukti}}, a
type-checker for $\l\Pi/\cR$ that we will use in our examples. The
code is available in \cite{sizechangetool} and could be easily adapted to a
subset of other languages like \tool{Agda}. As far as we
know, this tool is the first one to automatically check termination in
$\l\Pi/\cR$, which includes both higher-order rewriting and dependent
types.

This criterion is based on dependency pairs,
an important concept in the termination of first-order term rewriting
systems. It generalizes the
notion of recursive call in first-order functional
programs to rewriting. Namely, the dependency pairs of a rewriting rule
$f(l_1,\ldots,l_p)\a r$ are the pairs
$(f(l_1,\ldots,l_p),g(m_1,\ldots,m_q))$ such that $g(m_1,\ldots,m_q)$
is a subterm of $r$ and $g$ is a function symbol defined by some
rewriting rules. Dependency pairs have been introduced by Arts and Giesl
\cite{arts00tcs} and have evolved into a general framework for
termination \cite{giesl04lpar}. It is now at the heart of many
state-of-the-art automated termination provers for first-order rewriting systems and \tool{Haskell}, \tool{Java} or \tool{C} programs.

\hide{Indeed, Arts and Giesl
proved that a rewriting relation terminates if and only if there is no
infinite sequence of dependency pairs interleaved with reductions in
the arguments. In a first-order functional setting, it amounts to saying
that there is no infinite sequence of function calls.}

Dependency pairs have been extended to different simply-typed settings for
higher-order rewriting: Combinatory Reduction Systems
\cite{klop93tcs}\hide{kop12phd} and Higher-order Rewriting Systems
\cite{mayr98tcs}, with two different approaches: dynamic
dependency pairs include variable applications \cite{kop12phd}, while static
dependency pairs exclude them by slightly restricting the class of
systems that can be considered \cite{kusakari07aaecc}\hide{suzuki11pro}. Here, we use the static approach.

In \cite{wahlstedt07phd}, Wahlstedt considered a system slightly less general than $\l\Pi/\cR$ for which he provided conditions that imply the weak
normalization, that is, the existence of a finite reduction
to normal form.
In his system, $\cR$ uses matching on constructors only, like
in the languages \tool{OCaml} or \tool{Haskell}. In this case, $\cR$ is orthogonal: rules
are left-linear (no variable occurs twice in a left-hand side) and
have no critical pairs (no two rule left-hand side instances
overlap).
Wahlstedt's proof proceeds in two modular steps. First, he proves that
typable terms have a normal form if there is no infinite sequence of
function calls. Second, he proves that there is no infinite sequence
of function calls if $\cR$ satisfies Lee, Jones and Ben-Amram's
size-change termination criterion (SCT) \cite{lee01popl}.

In this paper, we extend Wahlstedt's results in two directions. First, we prove a
stronger normalization property: the absence of infinite
reductions. Second, we assume that $\cR$ is locally confluent,
a much weaker condition than orthogonality: rules can
be non-left-linear and have joinable critical pairs.

In \cite{blanqui05mscs}, the first author developed a termination
criterion for a calculus slightly more general than $\l\Pi/\cR$, based
on the notion of computability closure, assuming that type-level
rules are orthogonal. The computability closure of a term
$f(l_1,\ldots,l_p)$ is a set of terms that terminate whenever
$l_1,\ldots,l_p$ terminate. It is defined inductively thanks to
deduction rules preserving this property, using a precedence and a fixed well-founded ordering for dealing with function calls. Termination can then be
enforced by requiring each rule right-hand side to belong to the
computability closure of its corresponding left-hand side.

\hide{Function
calls are added in the computability closure by using a fixed
well-founded quasi-ordering. In a way similar to the recursive path
ordering \cite{dershowitz79focs}, a function call $g(m_1,\ldots,m_q)$
is added in the computability closure of $f(l_1,\ldots,l_p)$ if
$m_1,\ldots,m_q$ are in the computability closure and, either $g$ is
smaller than $f$ in a well-founded precedence, or $g$ is equivalent to
$f$ but $m_1,\ldots,m_q$ is smaller than $l_1,\ldots,l_p$ in the
multiset or lexicographic extension of some fixed well-founded ordering.
}

We extend this work as well by replacing that fixed ordering by the
dependency pair relation. In \cite{blanqui05mscs}, there must be a
decrease in every function call.
Using dependency pairs
allows one to have non-strict decreases. Then, following Wahlstedt, SCT can be used to
enforce the absence of infinite sequence of dependency pairs. But other
criteria have been developed for this purpose that could be adapted to $\l\Pi/\cR$.

\hide{
However, because in our setting typing depends on rewriting,
some dependency pairs may arise from
typing. For instance, if $f$ takes an argument of type $Ft$, then we will have dependency pairs between $f$ and $F$, even though the rules
of $f$ do not use $F$ itself. Our extended notion of dependency pair
boils down to the usual one when one restricts himself to simply-typed
rewriting rules though.
}

\subsubsection*{Outline}
The main result is Theorem \ref{thm-dp} stating that, for a large
class of rewriting systems $\cR$, the combination of $\b$
and $\cR$ is strongly normalizing on terms typable in $\l\Pi/\cR$ if,
roughly speaking, there is no infinite sequence of dependency pairs.

The proof involves two steps.
First, after recalling the terms and types
of $\l\Pi/\cR$ in Section \ref{sec-terms},
we introduce in Section
\ref{sec-interp} a model of this calculus based on
Girard's reducibility candidates \cite{girard88book}, and prove that every typable term is strongly normalizing if every symbol of the signature is in the interpretation of its type (Adequacy lemma).
Second, in Section \ref{sec-dp-thm}, we introduce our notion of dependency pair and prove that every symbol of the signature is in the interpretation of its type if there is no infinite sequence of dependency pairs.

In order to show the usefulness of this result, we give simple criteria for checking the conditions of the theorem.
In Section \ref{sec-acc}, we show that \emph{plain function passing} systems belong to the class of systems that we consider.
And in Section \ref{sec-sct}, we show how to use
size-change termination to obtain the termination of the dependency
pair relation.

Finally, in Section \ref{sec-comp} we compare our criterion with other criteria
and tools and, in Section \ref{sec-conclu}, we summarize our
results and give some hints on possible extensions.

For lack of space, some proofs are given in an appendix at the end of the paper.

\hide{
The paper is self-contained except for a few meta-theoretical
properties of $\l\Pi/\cR$ taken from \cite{blanqui01phd}.
The proofs of some lemmas are given in an appendix at the end of the paper.
Section
\ref{sec-interp} is the most technical part of the paper and can be
skipped at first. It introduces a new interpretation of types
(Definition \ref{def-interp}), much simpler than the one of Wahlstedt
\cite[Definition 3.2.3]{wahlstedt07phd}. In particular, it does not
use transfinite ordinal theory but a powerful fixpoint theorem instead.
}

%%% Local Variables:
%%%   TeX-master: "main"
%%%   mode: latex
%%%   mode: flyspell
%%%   ispell-local-dictionary: "english"
%%% End:

\section{Terms and types}
\label{sec-terms}

The set $\bT$ of terms of $\l\Pi/\cR$ is the same as those of Barendregt's $\l P$ \cite{barendregt92chapter}:
\begin{center}$t\in\bT = s\in\bS\mid x\in\bV\mid f\in\bF\mid\prod xtt\mid tt\mid\abs xtt$\end{center}
where $\bS=\{\type,\kind\}$ is the set of sorts\footnote{Sorts refer here to the notion of sort in Pure Type Systems, not the one used in some first-order settings.},
$\bV$ is an infinite set
of variables and $\bF$ is a set of function symbols, so that $\bS$,
$\bV$ and $\bF$ are pairwise disjoint.

Furthermore, we assume given a set $\cR$ of rules $l\a r$ such that
$\FV(r)\sle\FV(l)$ and
$l$ is of the form $f\vl$.
A symbol $f$ is said to be defined if there is a
rule of the form $f\vl\a r$. In this paper, we are interested in the
termination of
\begin{center}${\a}={\ab\cup\ar}$\end{center}
where $\ab$ is the $\b$-reduction of $\l$-calculus and
$\ar$ is the smallest
relation containing $\cR$ and closed by substitution and context:
we
consider rewriting with syntactic matching only. Following
\cite{blanqui16tcs}, it should however be possible to extend the
present results to rewriting with matching modulo $\b\eta$ or some
equational theory.
Let $\SN$ be the set of terminating terms and, given a term $t$, let $\red{t}=\{u\in\mb{T}\mid t\a u\}$ be the set of
immediate reducts of $t$.

A typing environment $\G$ is a (possibly empty) sequence
$x_1:T_1,\ldots,x_n:T_n$ of pairs of variables and terms, where the variables are distinct,
written $\vx:\vT$ for short. Given an environment $\G=\vx:\vT$ and a term $U$, let
$\all\G,U$ be $\prod\vx\vT U$.
The product arity $\arit(T)$ of a term $T$ is the
integer $n\in\bN$ such that $T=\prod{x_1}{T_1}\ldots\prod{x_n}{T_n}U$ and $U$ is
not a product. Let $\vt$ denote a possibly empty sequence of terms $t_1,\ldots,t_n$
of length $|\vt|=n$, and $\FV(t)$ be the set of free variables of $t$.

For each $f\in\bF$, we assume given a term $\T_f$ and a
sort $s_f$, and let $\G_f$ be the environment such that
$\T_f=\all\G_f,U$ and $|\G_f|=\arit(\T_f)$.

The application of a substitution $\s$ to a term $t$ is written $t\s$.
Given a substitution $\s$,
let $\dom(\s)=\{x|x\s\neq x\}$, $\FV(\s)=\bigcup_{x\in\dom(\s)}\FV(x\s)$ and
$[x\to a,\s]$ ($[x\to a]$ if $\s$ is the identity) be the substitution $\{(x,a)\}\cup\{(y,b)\in\s\mid
y\neq x\}$.
Given another substitution $\s'$,
let $\s\a\s'$ if there is $x$ such that $x\s\a x\s'$ and, for all $y\neq
x$, $y\s=y\s'$.

The typing rules of $\l\Pi/\cR$, in Figure
\ref{fig-typ}, add to those of $\l P$ the rule
(fun) similar to (var). Moreover, (conv) uses $\ad$ instead of
$\ad_\b$, where ${\ad}={\a^*{}^*\hspace{-3pt}\la}$ is the joinability
relation and $\a^*$ the reflexive and transitive closure of $\a$. We
say that $t$ has type $T$ in $\G$ if $\G\th t:T$ is derivable. A
substitution $\s$ is well-typed from $\D$ to $\G$, written
$\G\th\s:\D$, if, for all $(x:T)\in\D$, $\G\th x\s:T\s$ holds.

The word ``type'' is used to denote a term occurring at the right-hand side of a colon in a typing judgment (and we usually use capital letters for types). 
Hence, $\kind$ is the type of $\type$, $\T_f$ is the type of $f$,
and $s_f$ is the type of $\T_f$.
Common data types like natural numbers $\bN$ are usually declared in $\l\Pi$ as function symbols of type $\type$: $\T_\bN=\type$ and $s_\bN=\kind$.

The dependent product $\prod xAB$ generalizes the arrow type $A\A B$ of simply-typed $\l$-calculus: it is the type of
functions taking an argument $x$ of type $A$ and returning a term
whose type $B$ may depend on $x$.
If $B$ does not depend on $x$, we sometimes simply write $A\A B$.

Typing induces a hierarchy on terms \cite[Lemma
  47]{blanqui01phd}. At the top, there is the sort
$\kind$ that is not typable. Then, comes the class $\bK$ of kinds, whose
type is $\kind$: $K=\type\mid\prod xtK$ where $t\in\bT$. Then, comes the class of predicates,
whose types are kinds. Finally, at the bottom lie
(proof) objects whose types are predicates.

%%%%%%%%%%%%%%%%%%%%%%%%%%%%%%%%%%%%%%%%%%%%%%%%%%%%%%%%%%%%%%%%%%%%%%%%%%%%%%
\begin{figure}\caption{Typing rules of $\l\Pi/\cR$\label{fig-typ}}
\begin{center}
  \begin{tabular}{rc}
  (ax) & $\cfrac{}{\th\type:\kind}$\\[3mm]

  (var) & $\cfrac{\G\th A:s\quad x\notin\dom(\G)}{\G,x:A\th x:A}$\\[3mm]

  (weak) & $\cfrac{\G\th A:s\quad \G\th b:B\quad x\notin\dom(\G)}{\G,x:A\th b:B}$\\[3mm]

  (prod) & $\cfrac{\G\th A:\type\quad \G,x:A\th B:s}{\G\th\prod xAB:s}$\\
  \end{tabular}
  \begin{tabular}{rc}
  (abs) & $\cfrac{\G,x:A\th b:B\quad \G\th\prod xAB:s}{\G\th\l x:A.b:\prod xAB}$\\[3mm]

  (app) & $\cfrac{\G\th t:\prod xAB\quad \G\th a:A}{\G\th ta:B[x\to a]}$\\[3mm]

  (conv) & $\cfrac{\G\th a:A\quad A\ad B\quad \G\th B:s}{\G\th a:B}$\\[3mm]

  (fun) & $\cfrac{\th\T_f:s_f}{\th f:\T_f}$\\
  \end{tabular}
\end{center}
\end{figure}

%%%%%%%%%%%%%%%%%%%%%%%%%%%%%%%%%%%%%%%%%%%%%%%%%%%%%%%%%%%%%%%%%%%%%%%%%%%%%%
\begin{example}[Filter function on dependent lists]\label{expl-list-poly}
To illustrate the kind of systems we consider, we give an extensive
example in the new \tool{Dedukti} syntax combining type-level
rewriting rules ({\tt El} converts datatype codes into \tool{Dedukti}
types), dependent types ($\bL$ is the polymorphic type of lists
parameterized with their length), higher-order variables ({\tt fil} is
a function filtering elements out of a list along a boolean function
{\tt f}), and matching on defined function symbols ({\tt fil} can
match a list defined by concatenation). Note that this example cannot be represented in
\tool{Coq} or \tool{Agda} because of the rules using matching on {\tt
  app}. And its termination can be handled neither by
\cite{wahlstedt07phd} nor by \cite{blanqui05mscs} because the system
is not orthogonal and has no strict decrease in every recursive
call. It can however be handled by our new termination criterion and
its implementation \cite{sizechangetool}.
For readability, we removed the {\tt \&} which are used to identify pattern variables in the rewriting rules.
 % ⇒ \A ∀ \all → \ra
 \begin{lstlisting}[mathescape=true]
symbol Set: TYPE     symbol arrow: Set $\A$ Set $\A$ Set

symbol El: Set $\A$ TYPE          rule El (arrow a b) $\alp$ El a $\A$ El b

symbol Bool: TYPE      symbol true: Bool     symbol false: Bool
symbol Nat: TYPE       symbol zero: Nat      symbol s: Nat $\A$ Nat

symbol plus: Nat $\A$ Nat $\A $Nat        set infix 1 "+" $\coloneqq$ plus
  rule zero + q $\alp$ q               rule (s p) + q $\alp$ s (p + q)

symbol List: Set $\A$ Nat $\A$ TYPE
  symbol nil: $\all$a, List a zero
  symbol cons:$\all$a, El a $\A$ $\all$p, List a p $\A$ List a (s p)

symbol app: $\all$a p, List a p $\A$ $\all$q, List a q $\A$ List a (p+q)
  rule app a _ (nil _)        q m $\alp$ m
  rule app a _ (cons _ x p l) q m $\alp$ cons a x (p+q) (app a p l q m)

symbol len_fil: $\all$a, (El a $\A$ Bool) $\A$ $\all$p, List a p $\A$ Nat
symbol len_fil_aux: Bool $\A$ $\all$a, (El a $\A$ Bool) $\A$ $\all$p, List a p $\A$ Nat
  rule len_fil a f _ (nil _)         $\alp$ zero
  rule len_fil a f _ (cons _ x p l)  $\alp$ len_fil_aux (f x) a f p l
  rule len_fil a f _ (app _ p l q m)
       $\alp$ (len_fil a f p l) + (len_fil a f q m)
  rule len_fil_aux true  a f p l $\alp$ s (len_fil a f p l)
  rule len_fil_aux false a f p l $\alp$ len_fil a f p l

symbol fil: $\all$a f p l, List a (len_fil a f p l)
symbol fil_aux: $\all$b a f, El a $\A$ $\all$p l, List a (len_fil_aux b a f p l)
  rule fil a f _ (nil _)         $\alp$ nil a
  rule fil a f _ (cons _ x p l)  $\alp$ fil_aux (f x) a f x p l
  rule fil a f _ (app _ p l q m)
       $\alp$ app a (len_fil a f p l) (fil a f p l)
                (len_fil a f q m) (fil a f q m)
  rule fil_aux false a f x p l $\alp$ fil a f p l
  rule fil_aux true  a f x p l
       $\alp$ cons a x (len_fil a f p l) (fil a f p l)
\end{lstlisting}

\hide{Note that the rules of {\tt +} are required for the rules of {\tt
    app} to preserves typing.}
\end{example}

%%%%%%%%%%%%%%%%%%%%%%%%%%%%%%%%%%%%%%%%%%%%%%%%%%%%%%%%%%%%%%%%%%%%%%%%%%%%%%
\noindent{\bf Assumptions:} Throughout the paper, we assume that $\a$ is
locally confluent (${\la\a}\sle{\ad}$) and preserves typing (for all
$\G$, $A$, $t$ and $u$, if $\G\th t:A$ and $t\a u$, then $\G\th u:A$).

Note that local confluence implies that every $t\in\SN$ has a unique
normal form $\nf{t}$\hide{\cite{newman42am}}.

These assumptions are used in the interpretation of types (Definition
\ref{def-interp}) and the adequacy lemma (Lemma \ref{lem-comp}). Both
properties are undecidable in general. For confluence, \tool{Dedukti}
can call confluence checkers that understand the HRS format of the \href{http://project-coco.uibk.ac.at/}{confluence competition}. For
preservation of typing by reduction, it implements an heuristic
\cite{saillard15phd}.

\hide{See \cite{barbanera97jfp,blanqui05mscs,saillard15phd} for
  sufficient conditions.}

\hide{
Note that, when rules contain no abstraction and no variable
applications, then ${\a}={{\ab}\cup{\ar}}$ is locally confluent on
$\bT$ iff $\ar$ is locally confluent on first-order terms \cite[Lemma
  64]{blanqui01phd}.
}

\hide{
\begin{lemma}
$\G\th T:\kind$ iff $T\in\bK$ and $T$ is typable.
\end{lemma}

\begin{proof}
  We first prove that, if $\G\th t:T$ and $T=\kind$, then $T\in\bK$, by
  induction on $\G\th t:T$. We only detail important cases:
  \begin{itemize}
  \item[(var)] Then, $A=\kind$. This case is not possible since $\G\th
    A:s$ and $\kind$ is not typable.
  \item[(app)] If $B[x\to a]=\kind$, then either $B=\kind$ or else $B=x$ and
    $a=\kind$. Both cases are impossible since both $B$ and $a$ are
    typable while $\kind$ is not typable.
  \item[(conv)] This case is not possible since $\G\th B:s$ and $\kind$
    is not typable.
  \item[(fun)] Then, $\T_f=\kind$. This case is not possible since
    $\G\th\T_f:s_f$ and $\kind$ is not typable.
  \end{itemize}
  Hence, $T$ can only be typed by using (ax), (weak) and (prod).
\end{proof}
}

%%% Local Variables:
%%%   TeX-master: "main"
%%%   mode: latex
%%%   mode: flyspell
%%%   ispell-local-dictionary: "english"
%%% End:

\section{Interpretation of types as reducibility candidates}
\label{sec-interp}

We aim to prove the termination of the union of two relations, $\ab$
and $\ar$,\hide{sharing some symbols, namely the abstraction and the
application of $\l$-calculus,} on the set of well-typed terms (which
depends on $\cR$ since $\ad$ includes $\a_\mc{R}$). As is well known, termination is not modular in general.\hide{even when the two relations share no symbols}
\hide{\cite{toyama87ipl}} As a $\b$ step can generate an $\cR$ step, and
vice versa, we cannot expect to prove the termination of
${\ab}\cup{\ar}$ from the termination of $\ab$ and $\ar$.\hide{, unless $\cR$ is restricted to some class (e.g. object-level
first-order rewriting systems \cite{dougherty92ic,barthe98icalp}).}
The termination of $\l\Pi/\cR$ cannot be reduced to the termination of the
simply-typed $\l$-calculus either (as done for $\l\Pi$ alone in
\cite{harper93jacm}) because of type-level rewriting rules like the ones defining {\tt
  El} in Example \ref{expl-list-poly}. Indeed, type-level rules enable
the encoding of functional PTS like Girard's System F, whose
termination cannot be reduced to the termination of the simply-typed $\l$-calculus
\cite{cousineau07tlca}.

So, following Girard \cite{girard88book}, to prove the termination of $\ab\cup\ar$, we
build a model of our calculus by interpreting types into sets
of terminating terms. To this end, we need to find an interpretation
$\I~$ having the following properties:
\begin{itemize}
\item Because types are identified modulo conversion, we need $\I~$ to be
  invariant by reduction: if $T$ is typable and $T\a T'$, then
  we must have $\I{T}=\I{T'}$.
\item As usual, to handle $\b$-reduction, we need a product type
  $\prod xAB$ to be interpreted by the set of terms $t$ such
  that, for all $a$ in the interpretation of $A$, $ta$ is in the
  interpretation of $B[x\to a]$, that is, we must have $\I{\prod xAB}=\product{a}{\I{A}}{\I{B[x\to a]}}$ where $\product{a}{P}{Q(a)}=\{t\mid\all a\in P,ta\in
  Q(a)\}$.
\end{itemize}

First, we define the interpretation
of predicates (and $\type$) as the least fixpoint of
a monotone function in a directed-complete (= chain-complete) partial
order \cite{markowsky76au}. Second, we define the interpretation of
kinds by induction on their size.

%%%%%%%%%%%%%%%%%%%%%%%%%%%%%%%%%%%%%%%%%%%%%%%%%%%%%%%%%%%%%%%%%%%%%%%%%%%%%%
\begin{definition}[Interpretation of types]\label{def-interp}
  Let $\bI=\cF_p(\bT,\cP(\bT))$ be the set of partial functions from
  $\bT$ to the powerset of $\bT$.
  It is directed-complete wrt
  inclusion, allowing us to define $\cI$ as the least fixpoint of the monotone
  function $F:\bI\a\bI$ such that, if $I\in\bI$, then:
  \begin{itemize}
\item The domain of $F(I)$ is the set $D(I)$ of all the terminating
  terms $T$ such that, if $T$ reduces to some product term $\prod xAB$
  (not necessarily in normal form), then $A\in\dom(I)$ and, for all
  $a\in I(A)$, $B[x\to a]\in\dom(I)$.
\item If $T\in D(I)$ and the normal form\footnote{Because we assume
  local confluence, every terminating term $T$ has a unique normal
  form $\nf{T}$.} of $T$ is not a product, then $F(I)(T)=\SN$.
\item If $T\in D(I)$ and $\nf{T}=\prod xAB$, then
  $F(I)(T)=\product{a}{I(A)}{I(B[x\to a])}$.
\end{itemize}
  We now introduce $\cD=D(\cI)$ and
  define the interpretation of a term $T$ wrt to a
substitution $\s$, $\I{T}_\s$ (and simply $\I{T}$ if $\s$ is
the identity), as follows:
\begin{itemize}
\item $\I{s}_\s=\cD$ if $s\in\bS$,
\item $\I{\prod xAK}_\s=\product{a}{\I{A}_\s}{\I{K}_{[x\to a,\s]}}$ if
  $K\in\bK$ and $x\notin\dom(\s)$,
\item $\I{T}_\s=\cI(T\s)$ if $T\notin\bK\cup\{\kind\}$ and $T\s\in\cD$,
\item $\I{T}_\s=\SN$ otherwise.
\end{itemize}
A substitution $\s$ is adequate wrt an environment $\G$,
$\s\models\G$, if, for all $x:A\in\G$, $x\s\in\I{A}_\s$.
A typing map $\T$ is adequate if, for all $f$, $f\in\I{\T_f}$ whenever
$\th\T_f:s_f$ and $\T_f\in\I{s_f}$.

Let $\bC$ be the set of terms of the form $f\vt$ such that
$|\vt|=\arit(\T_f)$, $\th\T_f:s_f$, $\T_f\in\I{s_f}$ and, if
$\G_f=\vx:\vA$ and $\s=[\vx\to\vt]$, then $\s\models\G_f$.
(Informally, $\mb{C}$ is the set of terms obtained by fully applying
some function symbol to computable arguments.)
\end{definition}

We can then prove that, for all terms $T$, $\I{T}$ satisfies Girard's
conditions of reducibility candidates, called
computability predicates here, adapted to rewriting by including in
neutral terms every term of the form $f\vt$ when $f$ is applied to enough
arguments wrt $\cR$ \cite{blanqui05mscs}:

%%%%%%%%%%%%%%%%%%%%%%%%%%%%%%%%%%%%%%%%%%%%%%%%%%%%%%%%%%%%%%%%%%%%%%%%%%%%%%
\begin{definition}[Computability predicates]\label{def-comp}
A term is neutral if it is of the form ${(\abs xAt)u\vv}$, $x\vv$ or
$f\vv$ with, for every rule $f\vl\a r\in\cR$, $|\vl|\le|\vv|$.

Let $\bP$ be the set of all the sets of terms $S$ (computability
predicates) such that (a) $S\sle\SN$, (b) $\red{S}\sle S$, and (c)
$t\in S$ if $t$ is neutral and $\red{t}\sle S$.
\end{definition}

Note that neutral terms satisfy the following key property: if $t$ is
neutral then, for all $u$, $tu$ is neutral and every reduct of $tu$ is
either of the form $t'u$ with $t'$ a reduct of $t$, or of the form
$tu'$ with $u'$ a reduct of
$u$. \hide{${\red{tu}}={{\red{t}u}\cup{t\red{u}}}$}

One can easily check that $\SN$ is a computability predicate.

Note also that a computability predicate is never empty: it contains
every neutral term in normal form. In particular, it contains every
variable.

We then get the following results (the proofs are given in Appendix \ref{annex-interp}):

%%%%%%%%%%%%%%%%%%%%%%%%%%%%%%%%%%%%%%%%%%%%%%%%%%%%%%%%%%%%%%%%%%%%%%%%%%%%%%
\begin{lemma}\label{lem-props}
  \begin{enumerate}[(a)]
  \item\label{lem-comp-pred-int}
  For all terms $T$ and substitutions $\s$, $\I{T}_\s\in\bP$.
  \item\label{lem-int-red}
  If $T$ is typable, $T\s\in\cD$ and $T\a T'$, then
  $\I{T}_\s=\I{T'}_\s$.
  \item\label{lem-int-red-subs}
  If $T$ is typable, $T\s\in\cD$ and $\s\a\s'$, then
  $\I{T}_\s=\I{T}_{\s'}$.
  \item\label{lem-int-prod}
  If $\prod xAB$ is typable and $\prod x{A\s}{B\s}\in\cD$\hide{ and
  $x\notin\dom(\s)\cup\FV(\s)$},\\
  then $\I{\prod xAB}_\s=\product{a}{\I{A}_\s}{\I{B}_{[x\to a,\s]}}$.
  \item\label{lem-int-subs}
  If $\D\th U:s$, $\G\th\g:\D$ and $U\g\s\in\cD$, then
  $\I{U\g}_\s=\I{U}_{\g\s}$.
  \item\label{lem-comp-abs} Given $P\in\bP$ and, for all $a\in P$,
    $Q(a)\in\bP$ such that $Q(a')\sle Q(a)$ if $a\a a'$. Then, $\abs
    xAb\in\product{a}{P}{Q(a)}$ if $A\in\SN$ and, for all $a\in P$,
    $b[x\to a]\in Q(a)$.
  \end{enumerate}
\end{lemma}

We can finally prove that our model is adequate, that is, every term of type $T$
belongs to $\I{T}$, if the typing map $\T$ itself is adequate. This
reduces the termination of well-typed terms to the computability of
function symbols.

%%%%%%%%%%%%%%%%%%%%%%%%%%%%%%%%%%%%%%%%%%%%%%%%%%%%%%%%%%%%%%%%%%%%%%%%%%%%%%
\begin{lemma}[Adequacy]\label{lem-comp}
  If $\T$ is adequate, $\G\th t:T$ and $\s\models\G$, then
  $t\s\in\I{T}_\s$.
\end{lemma}

\begin{proof}
  First note that, if $\G\th t:T$, then either $T=\kind$ or $\G\th T:s$
  \cite[Lemma 28]{blanqui01phd}. Moreover, if $\G\th a:A$, $A\ad B$
  and $\G\th B:s$ (the premises of the (conv) rule), then $\G\th A:s$
  \cite[Lemma 42]{blanqui01phd} (because $\a$ preserves
  typing). Hence, the relation $\th$ is unchanged if one adds the
  premise $\G\th A:s$ in (conv), giving the rule (conv'). Similarly,
  we add the premise $\G\th\prod xAB:s$ in (app), giving the rule
  (app'). We now prove the lemma by induction on $\G\th t:T$ using
  (app') and (conv'):
  \begin{description}
  \item[(ax)]  It is immediate that $\type\in\I\kind_\s=\cD$.

  \item[(var)] By assumption on $\s$.

  \item[(weak)] If $\s\models\G,x:A$, then $\s\models\G$. So, the
    result follows by induction hypothesis.

  \item[(prod)] Is $(\prod xAB)\s$ in $\I{s}_\s=\cD$? Wlog we can assume
    $x\notin\dom(\s)\cup\FV(\s)$. So, ${(\prod xAB)\s}={\prod x{A\s}{B\s}}$. By
    induction hypothesis, $A\s\in\I\type_\s=\cD$. Let now $a\in
    \cI(A\s)$ and $\s'=[x\to a,\s]$. Note that $\cI(A\s)=\I{A}_\s$. So,
    $\s'\models\G,x:A$ and, by induction hypothesis,
    $B\s'\in\I{s}_\s=\cD$. Since $x\notin\dom(\s)\cup\FV(\s)$,
    we have $B\s'=(B\s)[x\to a]$. Therefore, $(\prod xAB)\s\in\I{s}_\s$.

  \item[(abs)] Is $(\abs xAb)\s$ in $\I{\prod xAB}_\s$? Wlog we can assume that
    $x\notin\dom(\s)\cup\FV(\s)$. So, ${(\abs xAb)\s}={\abs x{A\s}{b\s}}$. By
    Lemma \ref{lem-props}\ref{lem-int-prod}, $\I{\prod xAB}_\s=
    \product{a}{\I{A}_\s}{\I{B}_{[x\to a,\s]}}$. By Lemma
    \ref{lem-props}\ref{lem-int-red-subs}, $\I{B}_{[x\to a,\s]}$ is an
    $\I{A}_\s$-indexed family of computability predicates such that
    $\I{B}_{[x\to a',\s]}=\I{B}_{[x\to a,\s]}$ whenever $a\a
    a'$. Hence, by Lemma \ref{lem-props}\ref{lem-comp-abs}, $\abs
    x{A\s}{b\s}\in\I{\prod xAB}_\s$ if $A\s\in\SN$ and, for all
    $a\in\I{A}_\s$, $(b\s)[x\to a]\in\I{B}_{\s'}$ where $\s'=[x\to
      a,\s]$. By induction hypothesis,
    $(\prod xAB)\s\in\I{s}_\s=\cD$. Since
    $x\notin\dom(\s)\cup\FV(\s)$, $(\prod xAB)\s=\prod x{A\s}{B\s}$ and
    $(b\s)[x\to a]=b\s'$. Since $\cD\sle\SN$, we have
    $A\s\in\SN$. Moreover, since $\s'\models\G,x:A$, we have
    $b\s'\in\I{B}_{\s'}$ by induction hypothesis.

  \item[(app')] Is $(ta)\s=(t\s)(a\s)$ in $\I{B[x\to a]}_\s$? By induction
    hypothesis, $t\s\in\I{\prod xAB}_\s$, $a\s\in\I{A}_\s$ and
    $(\prod xAB)\s\in\I{s}=\cD$. By Lemma \ref{lem-props}\ref{lem-int-prod},
    $\I{\prod xAB}_\s=\product{\al}{\I{A}_\s}{\I{B}_{[x\to\al,\s]}}$. Hence,
    $(t\s)(a\s)\in\I{B}_{\s'}$ where $\s'=[x\to a\s,\s]$. Wlog we can
    assume $x\notin\dom(\s)\cup\FV(\s)$. So, $\s'=[x\to a]\s$. Hence,
    by Lemma \ref{lem-props}\ref{lem-int-subs}, $\I{B}_{\s'}=\I{B[x\to a]}_\s$.

  \item[(conv')] By induction hypothesis, $a\s\in\I{A}_\s$,
    $A\s\in\I{s}_\s=\cD$ and $B\s\in\I{s}_\s=\cD$. By Lemma
    \ref{lem-props}\ref{lem-int-red}, $\I{A}_\s=\I{B}_\s$. So, $a\s\in\I{B}_\s$.

  \item[(fun)] By induction hypothesis,
    $\T_f\in\I{s_f}_\s=\cD$. Therefore, $f\in\I{\T_f}_\s=\I{\T_f}$
    since $\T$ is adequate.\qedhere
  \end{description}
\end{proof}

\hide{As usual, by taking the identity for $\s$ (variables are computable),
we get that $\a$ terminates on terms typable in $\l\Pi/\cR$ if $\T$ is
adequate and $\a$ is locally confluent and preserves typing.}

%%% Local Variables:
%%%   TeX-master: "main"
%%%   mode: latex
%%%   mode: flyspell
%%%   ispell-local-dictionary: "english"
%%% End:

\section{Dependency pairs theorem}
\label{sec-dp-thm}

\newcommand\thfl{\th_{\!\!\!f\vl}}
\newcommand\thltf{\th_{\!\!\prec f}}
\newcommand{\cstr}{\Rrightarrow}

Now, we prove that the adequacy of $\T$
can be reduced to the absence of infinite sequences of dependency pairs, as
shown by Arts and Giesl for first-order rewriting \cite{arts00tcs}.

\begin{definition}[Dependency pairs]
Let $f\vl$ > $g\vm$ iff there is a rule $f\vl \a r \in \cR$, $g$ is defined
and $g\vm$ is a subterm of $r$ such that $\vm$ are all the arguments to which $g$ is applied.
The relation $>$ is the set of dependency pairs.

Let ${\call}={\a_\arg^*>_s}$ be the relation on the set $\bC$ (Def. \ref{def-interp}), where $f\vt\a_\arg f\vu$ iff $\vt\a_{prod}\vu$
(reduction in one argument), and $>_s$ is the closure by substitution
and left-application of $>$:
$f t_1\dots t_p~\call\,g u_1\dots u_q$ iff there are a dependency pair $f l_1\dots l_i>g m_1\dots m_j$ with $i\<p$ and $j\<q$
and a substitution $\s$ such that,
for all $k\<i$, $t_k\a^*l_k\s$ and, for all $k\<j$, $m_k\s=u_k$.
\end{definition}

In our setting, we have to close $>_s$ by left-application because
function symbols are curried. When a function symbol $f$ is not
fully applied wrt $\arit(\T_f)$, the missing arguments must be
considered as potentially being anything. Indeed, the following rewriting system:
\begin{lstlisting}[mathescape=true]
       app x y $\a$ x y             f x y $\a$ app (f x) y
\end{lstlisting}
whose dependency pairs are {\tt f x y > app (f x) y} and {\tt f x y > f x},
does not terminate, but there is no way to construct an infinite sequence of dependency pairs
without adding an argument to the right-hand side of the second dependency pair.

%%%%%%%%%%%%%%%%%%%%%%%%%%%%%%%%%%%%%%%%%%%%%%%%%%%%%%%%%%%%%%%%%%%%%%%%%%%%%%
\begin{example}\label{expl-mult}
 The rules of Example \ref{expl-list-poly} have the following dependency pairs (the pairs whose left-hand side is headed by {\tt fil} or {\tt fil\_aux} can be found in Appendix \ref{annex-matrices}):
\begin{lstlisting}[escapeinside={*}{*}]
*\color{red} A:*               El (arrow a b) > El a
*\color{red} B:*               El (arrow a b) > El b
*\color{red} C:*                    (s p) + q > p + q
*\color{red} D:*   app a _ (cons _ x p l) q m > p + q
*\color{red} E:*   app a _ (cons _ x p l) q m > app a p l q m
*\color{red} F:*len_fil a f _ (cons _ x p l)  > len_fil_aux (f x) a f p l
*\color{red} G:*len_fil a f _ (app _ p l q m) >
                                (len_fil a f p l) + (len_fil a f q m)
*\color{red} H:*len_fil a f _ (app _ p l q m) > len_fil a f p l
*\color{red} I:*len_fil a f _ (app _ p l q m) > len_fil a f q m
*\color{red} J:*    len_fil_aux true  a f p l > len_fil a f p l
*\color{red} K:*    len_fil_aux false a f p l > len_fil a f p l
\end{lstlisting}
\end{example}

In \cite{arts00tcs}, a sequence of dependency pairs interleaved with
$\a_\arg$ steps is called a chain. Arts and Giesl proved that, in a
first-order term algebra, $\ar$ terminates if and only if there are no
infinite chains, that is, if and only if $\call$ terminates.
Moreover, in a first-order term algebra, $\call$
terminates if and only if, for all $f$ and $\vt$, $f\vt$ terminates
wrt $\call$ whenever $\vt$ terminates wrt $\a$. In our framework,
this last condition is similar to saying that $\T$ is adequate.

We now introduce the class of systems to which we will extend
Arts and Giesl's theorem.

%%%%%%%%%%%%%%%%%%%%%%%%%%%%%%%%%%%%%%%%%%%%%%%%%%%%%%%%%%%%%%%%%%%%%%%%%%%%%%
\begin{definition}[Well-structured system]\label{def-wf-rule}\label{def-preced}
  Let $\succeq$ be the smallest quasi-order on $\bF$ such that
  $f\succeq g$ if $g$ occurs in $\T_f$ or if there is a rule $f\vl\a
  r\in\cR$ with $g$ (defined or undefined) occurring in $r$.
  Then, let ${\succ}={\succeq\moins\preceq}$ be the strict part of $\succeq$.
  A rewriting system $\cR$ is well-structured if:
  \begin{enumerate}[(a)]
  \item $\succ$ is well-founded;
  \item for every rule $f\vl\a r$, $|\vl|\le\arit(\T_f)$;
  \item for every dependency pair $f\vl>g\vm$, $|\vm|\leq \arit(\T_g)$;
  \item every rule $f\vl\a r$ is equipped with an environment $\D_{f\vl\a
    r}$ such that, if $\T_f={\prod\vx\vT U}$ and $\pi=[\vx\to\vl]$, then
    $\D_{f\vl\a r}\thfl r:U\pi$, where $\thfl$ is the restriction of
    $\th$ defined in Fig. \ref{fig-thf}.
  \end{enumerate}
\end{definition}

%%%%%%%%%%%%%%%%%%%%%%%%%%%%%%%%%%%%%%%%%%%%%%%%%%%%%%%%%%%%%%%%%%%%%%%%%%%%%%
\begin{figure}[t]
\caption{Restricted type systems $\thfl$ and $\thltf$ \label{fig-thf}}
\begin{center}
  \begin{tabular}{rc}
  (ax) & $\cfrac{}{\thfl\type:\kind}$\\[3mm]

  (var) & $\cfrac{\G\thltf A:s\quad x\notin\dom(\G)}{\G,x:A\thfl x:A}$\\[3mm]
  \end{tabular}
  \begin{tabular}{rc}
  (weak) & $\cfrac{\G\thltf A:s\quad \G\thfl b:B\quad x\notin\dom(\G)}{\G,x:A\thfl b:B}$\\[3mm]

  (prod) & $\cfrac{\G\thfl A:\type\quad \G,x:A\thfl B:s}{\G\thfl\prod xAB:s}$\\[3mm]
  \end{tabular}
  \begin{tabular}{rc}
  (abs) & $\cfrac{\G,x:A\thfl b:B\quad \G\thltf\prod xAB:s}{\G\thfl\abs xAb:\prod xAB}$\\[3mm]

  (app') &
   $\cfrac{\G\thfl t:\prod xAB\quad \G\thfl a:A \quad \G\thltf \prod xAB:s}
   {\G\thfl ta:B[x\to a]}$\\[3mm]

  (conv') &
  $\cfrac{\G\thfl a:A\quad A\ad B\quad \G\thltf B:s \quad \G\thltf A:s}
   {\G\thfl a:B}$\\[3mm]
  \end{tabular}
  \begin{tabular}{rcl}

  (dp) & $\cfrac{\thltf\T_g:s_g\quad \G\thfl \g:\S}{\G\thfl g\vy\g:V\g}$ &
  ($\T_g=(\prod\vy\vU V),\S=\vy:\vU,g\vy\g<f\vl$)\\[3mm]

  (const) & $\cfrac{\thltf \T_g:s_g}{\thfl g:\T_g}$ &
  ($g$ undefined)\\
  \end{tabular}
\end{center}
and $\thltf$ is defined by the same rules as $\th$, except (fun) replaced by:
\begin{center}
  \begin{tabular}{rc}
  (fun$_{\prec f}$) & $\cfrac{\thltf\T_g:s_g\quad g\prec f}{\thltf g:\T_g}$
  \end{tabular}
\end{center}
\end{figure}

Condition (a) is always satisfied when $\bF$ is finite.
Condition (b) ensures that a term of the form $f\vt$ is neutral
whenever $|\vt|=\arit(\T_f)$. Condition (c) ensures that $>$ is included
in $\call$.

The relation $\thfl\,$ corresponds to the notion of computability
closure in \cite{blanqui05mscs}, with the ordering on function calls
replaced by the dependency pair relation. It is similar to $\th$
except that it uses the variant of (conv) and (app) used in the proof
of the adequacy lemma; (fun) is split in the rules (const) for
undefined symbols and (dp) for dependency pairs whose left-hand side
is $f\vl$; every type occurring in an object term or every type of a
function symbol occurring in a term is required to be typable by using
symbols smaller than $f$ only.

The environment $\D_{f\vl\a r}$ can be inferred by \tool{Dedukti} when
one restricts rule left hand-sides to some well-behaved class of terms
like algebraic terms or Miller patterns\hide{\cite{miller91jlc}} (in $\l$Prolog).

One can check that Example \ref{expl-list-poly} is
well-structured (the proof is given in Appendix \ref{annex-matrices}).

Finally, we need matching to be compatible with computability, that
is, if $f\vl\a r\in\cR$ and $\vl\s$ are computable, then $\s$ is
computable, a condition called accessibility in \cite{blanqui05mscs}:

%%%%%%%%%%%%%%%%%%%%%%%%%%%%%%%%%%%%%%%%%%%%%%%%%%%%%%%%%%%%%%%%%%%%%%%%%%%%%%
\begin{definition}[Accessible system]\label{def-valid-rule}
  A well-structured system $\cR$ is accessible if, for all
  substitutions $\s$ and rules $f\vl\a r$ with $\T_f=\prod\vx\vT U$ and
  $|\vx|=|\vl|$, we have $\s\models\D_{f\vl\a r}$ whenever
  $\th\T_f:s_f$, $\T_f\in\I{s_f}$ and
  $[\vx\mapsto\vl]\s\models\vx:\vT$.
\end{definition}

This property is not always satisfied because the subterm relation
does not preserve computability in general. Indeed, if $C$ is an undefined type
constant, then $\I{C}=\SN$. However, $\I{C\A C}\neq\SN$ since $\w=\abs
xC{xx}\in\SN$ and $\w\w\notin\SN$. Hence, if $c$ is an undefined function symbol
of type $\T_c=(C\A C)\A C$, then $c\,\w\in\I{C}$ but $\w\notin\I{C\A C}$.

We can now state the main lemma:

%%%%%%%%%%%%%%%%%%%%%%%%%%%%%%%%%%%%%%%%%%%%%%%%%%%%%%%%%%%%%%%%%%%%%%%%%%%%%%
\begin{lemma}\label{lem-comp-fun}
  $\T$ is adequate if $\call$ terminates and $\cR$ is well-structured
  and accessible.
\end{lemma}

\begin{proof}
  Since $\cR$ is well-structured, $\succ$ is well-founded by condition
  (a). We prove that, for all $f\in\bF$, $f\in\I{\T_f}$, by induction
  on $\succ$. So, let $f\in\bF$ with $\T_f=\all\G_f,U$ and
  $\G_f=x_1:T_1,\ldots,x_n:T_n$. By induction hypothesis, we have
  that, for all $g\prec f$, $g\in\I{\T_g}$.

  Since $\a_\arg$ and $\call$ terminate on $\bC$ and
  ${\a_\arg\call}\sle{\call}$, we have that ${\a_\arg}\cup{\call}$
  terminates. We now prove that, for all $f\vt\in\bC$, we have
  $f\vt\in\I{U}_\t$ where $\t=[\vx\to\vt]$, by a second induction on
  ${\a_\arg}\cup{\call}$. By condition (b),
  $f\vt$ is neutral. Hence, by definition of computability, it
  suffices to prove that, for all $u\in{\red{f\vt}}$,
  $u\in\I{U}_\t$. There are 2 cases:
  \begin{itemize}
  \item $u=f\vv$ with $\vt\a_{prod}\vv$. Then, we can conclude by
    the first induction hypothesis.

  \item There are $fl_1\ldots l_k\a r\in\cR$ and $\s$ such that
    $u=(r\s) t_{k+1}\ldots t_n$ and, for all $i\in\{1,\ldots,k\}$,
    $t_i=l_i\s$. Since $f\vt\in\bC$, we have $\pi\s\models\G_f$. Since
    $\cR$ is accessible, we get that $\s\models\D_{f\vl\a r}$. By
    condition (d), we have $\D_{f\vl\a r}\thfl r:V\pi$ where
    $V=\prod{x_{k+1}}{T_{k+1}}\ldots\prod{x_n}{T_n}U$.

    Now, we prove that, for all $\G$, $t$ and $T$, if $\G\thfl t:T$
    ($\G\thltf t:T$ resp.) and $\s\models\G$, then $t\s\in\I{T}_\s$,
    by a third induction on the structure of the derivation of $\G\thfl t:T$
    ($\G\thltf t:T$ resp.), as in the proof of Lemma \ref{lem-comp}
    except for (fun) replaced by (fun$_{\prec f}$) in one case, and
    (const) and (dp) in the other case.
    \begin{description}
     \item[(fun$_{\prec f}$)] We have $g\in\I{\T_g}$ by the first induction
       hypothesis on $g$.
     \item[(const)] Since $g$ is undefined, it is neutral and normal.
       Therefore, it belongs to every computability predicate and in
       particular to $\I{\T_g}_\s$.
     \item[(dp)] By the third induction hypothesis, $y_i\g\s\in\I{U_i\g}_\s$. By
       Lemma \ref{lem-props}\ref{lem-int-subs},
       $\I{U_i\g}_\s=\I{U_i}_{\g\s}$. So, $\g\s\models\S$ and
       $g\vy\g\s\in\bC$. Now, by condition (c),
       $g\vy\g\s\tilde{<}f\vl\s$ since $g\vy\g<f\vl$. Therefore, by
       the second induction hypothesis, $g\vy\g\s\in\I{V\g}_\s$.
   \end{description}
   So, $r\s\in\I{V\pi}_\s$ and, by Lemma \ref{lem-props}\ref{lem-int-prod},
   $u\in\I{U}_{[x_n\to t_n,..,x_{k+1}\to t_{k+1},\pi\s]}=\I{U}_\t$.\qedhere
  \end{itemize}
\end{proof}

Note that the proof still works if one replaces the relation $\succeq$ of Definition \ref{def-preced} by any well-founded quasi-order such that $f\succeq g$ whenever $f\vl>g\vm$.
The quasi-order of Definition \ref{def-preced}, defined syntactically, relieves the user of the burden of providing one and is sufficient in every practical case met by the authors.
However it is possible to construct ad-hoc systems which require a quasi-order richer than the one presented here.

By combining the previous lemma and the Adequacy lemma (the identity
substitution is computable), we get the main result of the
paper:

%%%%%%%%%%%%%%%%%%%%%%%%%%%%%%%%%%%%%%%%%%%%%%%%%%%%%%%%%%%%%%%%%%%%%%%%%%%%%%
\begin{theorem}\label{thm-dp}
  The relation ${\a}={{\ab}\cup{\ar}}$ terminates on terms typable in
  $\l\Pi/\cR$ if $\a$ is locally confluent and preserves typing, $\cR$
  is well-structured and accessible, and $\call$ terminates.
\end{theorem}

For the sake of completeness, we are now going to give sufficient
conditions for accessibility and termination of $\call$ to hold, but
one could imagine many other criteria.

%%% Local Variables:
%%%   TeX-master: "main"
%%%   mode: latex
%%%   mode: flyspell
%%%   ispell-local-dictionary: "english"
%%% End:

\section{Checking accessibility}
\label{sec-acc}

In this section, we give a simple condition to ensure accessibility
and some hints on how to modify the interpretation when this condition is
not satisfied.

As seen with the definition of accessibility, the main problem is to
deal with subterms of higher-order type. A simple condition is
to require higher-order variables to be direct subterms of
the left-hand side, a condition called plain function-passing (PFP) in
\cite{kusakari07aaecc}, and satisfied by Example \ref{expl-list-poly}.

\begin{definition}[PFP systems]\label{def-pfp-rule}
A well-structured $\cR$ is PFP if, for all $f\vl\a r\in\cR$ with
$\T_f=\prod\vx\vT U$ and $|\vx|=|\vl|$, $\vl\notin\bK\cup\{\kind\}$ and, for
all $y:T\in\D_{f\vl\a r}$, there is $i$ such that $y=l_i$ and
$T=T_i[\vx\to\vl]$, or else $y\in\FV(l_i)$ and $T=D\vt$ with $D$
undefined and $|\vt|=\arit(D)$.
\end{definition}

%%%%%%%%%%%%%%%%%%%%%%%%%%%%%%%%%%%%%%%%%%%%%%%%%%%%%%%%%%%%%%%%%%%%%%%%%%%%%%
\begin{lemma}\label{lem-pfp-rule}
PFP systems are accessible.
\end{lemma}

\begin{proof}
  Let $f\vl\a r$ be a PFP rule with $\T_f=\all\G,U$, $\G=\vx:\vT$,
  $\pi={[\vx\to\vl]}$. Following Definition \ref{def-valid-rule}, assume
  that $\th\T_f:s_f$, $\T_f\in\cD$ and $\pi\s\models\G$. We have
  to prove that, for all $(y:T)\in\D_{f\vl\a r}$, $y\s\in\I{T}_\s$.
  \begin{itemize}
  \item Suppose $y=l_i$ and $T=T_i\pi$. Then,
    $y\s=l_i\s\in\I{T_i}_{\pi\s}$. Since $\th\T_f:s_f$,
    $T_i\notin\bK\cup\{\kind\}$. Since $\T_f\in\cD$ and
    $\pi\s\models\G$, we have $T_i\pi\s\in\cD$. So,
    $\I{T_i}_{\pi\s}=\cI(T_i\pi\s)$. Since $T_i\notin\bK\cup\{\kind\}$ and
    $\vl\notin\bK\cup\{\kind\}$, $T_i\pi\notin\bK\cup\{\kind\}$. Since
    $T_i\pi\s\in\cD$, $\I{T_i\pi}_\s=\cI(T_i\pi\s)$. Thus,
    $y\s\in\I{T}_\s$.
  \item Suppose $y\in\FV(l_i)$ and $T$ is of the form $C\vt$ with
    $|\vt|=\arit(C)$. Then, $\I{T}_\s=\SN$ and $y\s\in\SN$ since
    $l_i\s\in\I{T_i}_\s\sle\SN$.\qedhere
  \end{itemize}
\end{proof}

But many accessible systems are not PFP. They can be
proved accessible by changing the interpretation of type constants (a
complete development is left for future work).

%%%%%%%%%%%%%%%%%%%%%%%%%%%%%%%%%%%%%%%%%%%%%%%%%%%%%%%%%%%%%%%%%%%%%%%%%%%%%%
\begin{example}[Recursor on Brouwer ordinals]\label{expl-ord}\hfill
 \begin{lstlisting}[mathescape=true]
symbol Ord: TYPE
 symbol zero: Ord  symbol suc: Ord$\A$Ord     symbol lim: (Nat$\A$Ord)$\A$Ord

symbol ordrec: A$\A$(Ord$\A$A$\A$A)$\A$((Nat$\A$Ord)$\A$(Nat$\A$A)$\A$A)$\A$Ord$\A$A
  rule ordrec u v w zero    $\alp$ u
  rule ordrec u v w (suc x) $\alp$ v x (ordrec u v w x)
  rule ordrec u v w (lim f) $\alp$ w f ($\lambda$n,ordrec u v w (f n))
\end{lstlisting}
\end{example}

  The above example is not PFP because {\tt f:Nat$\A$Ord} is not
  argument of {\tt ordrec}.
Yet, it is accessible if one takes for $\I{\tt Ord}$ the least fixpoint
of the monotone function $F(S)=\{{t\in\SN\mid} \text{if }t\a^*{\tt lim}\,f\text{ then }
f\in\I{\tt Nat}\A S\text{, and if } t\a^*{\tt suc}\,u\text{ then } u\in S\}$ \cite{blanqui05mscs}.

Similarly, the following encoding of the simply-typed $\l$-calculus is
not PFP but can be proved accessible by taking
\[\I{{\tt T}~
  c}=\text{if }\nf{c}={\tt arrow}\,a\,b\text{ then }\{t\in\SN\mid \text{if }t\a^*{\tt
  lam} f\text{ then }f\in\I{{\tt T}~a}\A\I{{\tt T}~b}\}\text{ else }\SN\]

\begin{example}[Simply-typed $\l$-calculus]\label{expl-lambda}\hfill
\begin{lstlisting}[mathescape=true]
symbol Sort : TYPE              symbol arrow : Sort $\A$ Sort $\A$ Sort

symbol T : Sort $\A$ TYPE
  symbol lam : $\forall$ a b, (T a $\A$ T b) $\A$ T (arrow a b)
  symbol app : $\forall$ a b, T (arrow a b) $\A$ T a $\A$ T b
  rule app a b (lam _ _ f) x $\alp$ f x
\end{lstlisting}
\end{example}

\hide{
Note here, that erasing dependencies in types, would lead to an encoding of
\emph{pure} $\lambda$-calculus, which is well-known to be non-normalizing.
}

%%% Local Variables:
%%%   TeX-master: "main"
%%%   mode: latex
%%%   mode: flyspell
%%%   ispell-local-dictionary: "english"
%%% End:

\section{Size-change termination}
\label{sec-sct}

In this section, we give a sufficient condition for $\call$ to terminate.
For first-order rewriting,
many techniques have been developed for that purpose. To cite just a
few, see for instance \cite{hirokawa07ic,giesl06jar}.\hide{thiemann07phd}
Many of them can probably be extended to $\l\Pi/\cR$,
either because the structure of terms in which they are expressed can be
abstracted away, or because they can be extended to deal also with
variable applications, $\l$-abstractions and $\b$-reductions.
\hide{\cite{suzuki11pro,kop12phd}}

As an example, following Wahlstedt \cite{wahlstedt07phd}, we are going
to use Lee, Jones and Ben-Amram's size-change termination criterion
(SCT) \cite{lee01popl}. It consists in following arguments along
function calls and checking that, in every potential loop, one of them
decreases. First introduced for first-order functional languages, it
has then been extended to many other settings: untyped $\l$-calculus
\cite{jones04rta}, a subset of \tool{OCaml} \cite{sereni05aplas},
Martin-L\"of's type theory \cite{wahlstedt07phd}, System F
\cite{lepigre17draft}.

We first recall Hyvernat and Raffalli's matrix-based presentation of
SCT \cite{hyvernat10wst}:

\newcommand{\arity}[1]{\arit(\T_#1)}

%%%%%%%%%%%%%%%%%%%%%%%%%%%%%%%%%%%%%%%%%%%%%%%%%%%%%%%%%%%%%%%%%%%%%%%%%%%%%%
\begin{definition}[Size-change termination]\label{def-sct}
  Let $\tgt$ be the smallest transitive relation such that $ft_1\dots
  t_n\tgt t_i$ when $f\in\bF$. The call graph $\cG(\cR)$ associated to
  $\cR$ is the directed labeled graph on the defined symbols of $\bF$
  such that there is an edge between $f$ and $g$ iff there is a
  dependency pair $fl_1\dots l_p>gm_1\dots m_q$. This edge is labeled
  with the matrix $(a_{i,j})_{i\leq\arity{f},j\leq\arity{g}}$ where:
  \begin{itemize}
    \item if $l_i\tgt m_j$, then $a_{i,j}=-1$;
    \item if $l_i=m_j$, then $a_{i,j}=0$;
    \item otherwise $a_{i,j}=\infty$ (in particular if $i>p$ or $j>q$).
  \end{itemize}
  $\cR$ is size-change terminating (SCT) if, in the transitive closure
  of $\cG(\cR)$ (using the min-plus semi-ring to multiply the matrices
  labeling the edges), all idempotent matrices labeling a loop have
  some $-1$ on the diagonal.
\end{definition}

We add lines and columns of $\infty$'s in matrices associated to
dependency pairs containing partially applied symbols (cases $i>p$ or
$j>q$) because missing arguments cannot be
compared with any other argument since they are arbitrary.

The matrix associated to the dependency pair {\tt C: (s p) + q > p +
  q} and the call graph associated to the dependency pairs of Example
\ref{expl-mult} are depicted in Figure \ref{fig-callgraph}.  The full
list of matrices and the extensive call graph of Example
\ref{expl-list-poly} can be found in Appendix \ref{annex-matrices}.

%%%%%%%%%%%%%%%%%%%%%%%%%%%%%%%%%%%%%%%%%%%%%%%%%%%%%%%%%%%%%%%%%%%%%%%%%%%%%%
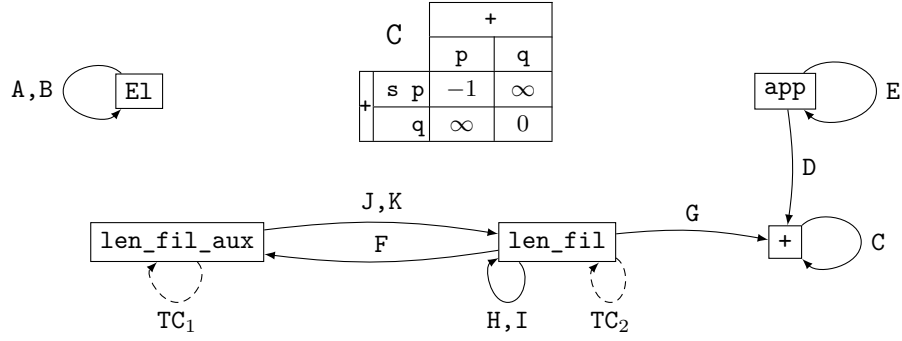
\begin{figure}\caption{Matrix of dependency pair {\tt C} and call graph of the dependency pairs of Example \ref{expl-mult}\label{fig-callgraph}}
\begin{center}\tt
\begin{tabular}{|@{}r@{}|@{}r@{\,}|c|c|}
    \cline{3-4}
    \multicolumn{2}{c|}{\multirow{2}{*}{\Large C}} & \multicolumn{2}{c|}{\texttt{+}}\\
    \cline{3-4}
    \multicolumn{2}{c|}{} & {\texttt{p}} & {\texttt{q}}\\
    \hline
      \multirow{2}{*}{\rotatebox{90}{\texttt{+}}}
      & {\texttt{\phantom{,}s p}} & $-1$ & $\infty$\\
    \cline{2-4}
      & {\texttt{q}} & $\infty$ & $0$\\
    \hline
\end{tabular}
\end{center}
\vspace{-5.5em}
\begin{center}\tt
  \begin{tikzpicture}
    \node[draw] (l_filter) at (2,-2) {len\_fil};
    \node[draw] (l_f_aux) at (-3,-2) {len\_fil\_aux};
    \node[draw] (el) at (-3.5,0) {El};
    \node[draw] (app) at (5,0) {app};
    \node[draw] (pl) at (5,-2) {+};
    \draw[>=latex,->] (el) to[out=135,in=90] (-4.5,0) to[out=-90,in=-135] (el);
    \node[left] (ell) at (-4.5,0) {A,B};
    \draw[>=latex,->] (pl) to[out=45,in=90] (6,-2) to[out=-90,in=-45] (pl);
    \node[right] (pll) at (6,-2) {C};
    \draw[>=latex,->] (app) to[bend left=8] node[midway,right] {D} (pl);
    \draw[>=latex,->] (app) to[out=45,in=90] (6.2,0) to[out=-90,in=-45] (app);
    \node[right] (appl) at (6.2,0) {E};
    \draw[>=latex,->] (l_filter) to[bend left=8] node[midway,above] {F} (l_f_aux);
    \draw[>=latex,->] (l_filter) to[bend left=8] node[midway,above] {G} (pl);
    \draw[>=latex,->] (1.5,-2.27) to[out=-45,in=0] (1.35,-2.8) to[out=180,in=-150] (1.23,-2.21);
    \node[below] (l_filterl) at (1.35,-2.8) {H,I};
    \draw[>=latex,->] (l_f_aux) to[bend left=8] node[midway,above] {J,K} (l_filter);
    \draw[>=latex,->,densely dashed] (l_f_aux) to[out=-45,in=0] (-3,-2.8) to[out=180,in=-135] (l_f_aux);
    \node[below] (l_f_auxTC) at (-3,-2.8) {TC${}_1$};
    \draw[>=latex,->,densely dashed] (2.77,-2.21) to[out=-30,in=0] (2.7,-2.8) to[out=180,in=-135] (2.55,-2.27);
    \node[below] (l_filterTC) at (2.7,-2.8) {TC${}_2$};
    % \draw[>=latex,->] (app) to[bend left=10] node[midway,right] {$\beta$} (el);
    % \draw[>=latex,->] (l_filter) to[bend left=10] node[midway,right] {$\delta$,$\epsilon$} (el);
    % \draw[>=latex,->] (l_f_aux) to[bend left=10] node[midway,right] {$\zeta$,$\eta$} (el);
    % \draw[>=latex,->] (filter) to[bend left=10] node[midway,right] {$\theta$,$\kappa$,$\nu$} (el);
    % \draw[>=latex,->] (filter) to node[pos=0.6,above] {$\mu$} (l_f_aux);
    % \draw[>=latex,->] (filter) to[bend left=10] node[midway,right] {$o$} (pl);
    % \draw[>=latex,->] (f_aux) to[bend left=10] node[midway,right] {$\pi$,$\tau$} (el);
    % \draw[>=latex,->] (f_aux) to node[midway,left] {$\sigma$,$\phi$} (l_f_aux);
  \end{tikzpicture}
\end{center}
\end{figure}

%%%%%%%%%%%%%%%%%%%%%%%%%%%%%%%%%%%%%%%%%%%%%%%%%%%%%%%%%%%%%%%%%%%%%%%%%%%%%%
\begin{lemma}\label{lem-sct}
  $\call$ terminates if $\bF$ is finite and $\cR$ is SCT.
\end{lemma}

\begin{proof}
 Suppose that there is an infinite sequence $\chi=f_1\vt_1\call
 f_2\vt_2\call\dots$ Then, there is an infinite path in the call
 graph going through nodes labeled by $f_1,f_2,\dots$ Since $\bF$ is
 finite, there is a symbol $g$ occurring infinitely often in this
 path. So, there is an infinite sequence $g\vu_1,g\vu_2,\dots$
 extracted from $\chi$. Hence, for every $i,j\in\bN^*$, there is a
 matrix in the transitive closure of the graph which labels the loops
 of $g$ corresponding to the relation between $\vu_i$ and
 $\vu_{i+j}$. By Ramsey's theorem\hide{\cite{ramsey30plms}}, there is an
 infinite sequence $(\phi_i)$ and a matrix $M$ such that $M$ corresponds
 to all the transitions $g\vu_{\phi_i},g\vu_{\phi_j}$ with
 $i\neq j$. $M$ is idempotent, indeed
 $g\vu_{\phi_i},g\vu_{\phi_{i+2}}$ is labeled by $M^2$ by
 definition of the transitive closure and by $M$ due to Ramsey's
 theorem, so $M=M^2$. Since, by hypothesis, $\cR$ satisfies SCT, there
 is $j$ such that $M_{j,j}$ is $-1$. So, for all $i$,
 $u^{(j)}_{\phi_{i}}(\a^*\tgt)^+ u^{(j)}_{\phi_{i+1}}$.
 Since ${\tgt\!\a}\subseteq{\a\!\tgt}$ and $\a_{\arg}$ is well-founded on $\bC$,
 the existence of an infinite sequence contradicts the fact that $\tgt$ is well-founded.
\end{proof}

By combining all the previous results, we get:

%%%%%%%%%%%%%%%%%%%%%%%%%%%%%%%%%%%%%%%%%%%%%%%%%%%%%%%%%%%%%%%%%%%%%%%%%%%%%%
\begin{theorem}\label{thm-sn}
  The relation ${\a}={{\ab}\cup{\ar}}$ terminates on terms typable in
  $\l\Pi/\cR$ if $\a$ is locally confluent and preserves typing, $\bF$
  is finite and $\cR$ is well-structured, plain-function passing and
  size-change terminating.
\end{theorem}

The rewriting system of Example \ref{expl-list-poly} verifies all these conditions (proof in the appendix).

%%% Local Variables:
%%%   TeX-master: "main"
%%%   mode: latex
%%%   mode: flyspell
%%%   ispell-local-dictionary: "english"
%%% End:

\section{Implementation and comparison with other criteria and tools}
\label{sec-comp}

We implemented our criterion in a tool called \tool{SizeChangeTool}
\cite{sizechangetool}. As far as we know, there are no other
termination checker for $\l\Pi/\cR$.

If we restrict ourselves to simply-typed rewriting systems, then we
can compare it with the termination checkers participating in the
category ``higher-order rewriting union beta'' of the \href{http://termination-portal.org/wiki/Termination\_Competition}{termination
competition}\hide{\cite{tc}}: \href{http://wandahot.sourceforge.net/}{\tool{Wanda}}\hide{\cite{wanda}} uses dependency
pairs, polynomial interpretations, HORPO and many transformation
techniques \cite{kop12phd}; \tool{SOL} uses the General Schema
\cite{blanqui16tcs} and other techniques\hide{; \tool{THOR} uses some semantic version of HORPO
  \cite{blanqui07jacm}}. As these tools implement various techniques
and \tool{SizeChangeTool} only one, it is difficult to compete with
them. Still, there are examples that are solved by
\tool{SizeChangeTool} and not by one of the other tools, demonstrating
that these tools would benefit from implementing our new technique. For
instance, the problem {\tt Hamana\_Kikuchi\_18/h17}\hide{\cite{tpdb}} is
proved terminating by \tool{SizeChangeTool} but not by \tool{Wanda}
because of the rule:
\begin{lstlisting}[mathescape=true]
 rule map f (map g l) $\ra$ map (comp f g) l
\end{lstlisting}
\noindent
And the problem {\tt Kop13/kop12thesis\_ex7.23} is proved terminating
by \tool{SizeChangeTool} but not by \tool{Sol} because of the
rules:\footnote{We renamed the function symbols for the sake of
  readability.}
\begin{lstlisting}[mathescape=true]
 rule f h x (s y) $\ra$ g (c x (h y)) y     rule g x y $\ra$ f ($\lambda$_,s 0) x y
\end{lstlisting}

One could also imagine to translate a termination problem in
$\l\Pi/\cR$ into a simply-typed termination problem. Indeed, the
termination of $\l\Pi$ alone (without rewriting) can be reduced to the
termination of the simply-typed $\l$-calculus
\cite{harper93jacm}. This has been extended to $\l\Pi/\cR$ when there
are no type-level rewrite rules like the ones defining {\tt El} in
Example \ref{expl-list-poly} \cite{jouannaud15tlca}.  However,
this translation does not preserve termination as shown by
the Example \ref{expl-lambda} which is not terminating if all the
types $\bT x$ are mapped to the same type constant.

In \cite{roux11rta}, Roux also uses dependency pairs for the
termination of simply-typed higher-order rewriting systems, as well as a
restricted form of dependent types where a type constant $C$ is
annotated by a pattern $l$ representing the set of terms matching
$l$. This extends to patterns the notion of indexed or sized types
\cite{hughes96popl}. Then, for proving the absence of infinite chains,
he uses simple projections \cite{hirokawa07ic}, which can be seen as a
particular case of SCT where strictly decreasing arguments are fixed
(SCT can also handle permutations in arguments).

Finally, if we restrict ourselves to orthogonal systems, it is also
possible to compare our technique to the ones implemented in the proof
assistants \tool{Coq} and \tool{Agda}. \tool{Coq} essentially
implements a higher-order version of primitive recursion. \tool{Agda}
on the other hand uses SCT.

Because Example \ref{expl-list-poly} uses matching on defined symbols,
it is not orthogonal and can be written neither in \tool{Coq} nor in
\tool{Agda}. \tool{Agda} recently added the possibility of adding
rewrite rules but this feature is highly experimental and comes with
no guaranty. In particular, \tool{Agda} termination checker does not
handle rewriting rules.

\tool{Coq} cannot handle inductive-recursive definitions
\cite{dybjer00jsl} nor function definitions with permuted arguments in
function calls while it is no problem for \tool{Agda} and us.

%%% Local Variables:
%%%   TeX-master: "main"
%%%   mode: latex
%%%   mode: flyspell
%%%   ispell-local-dictionary: "english"
%%% End:

\section{Conclusion and future work}
\label{sec-conclu}

We proved a general modularity result extending Arts and
Giesl's theorem that a rewriting relation terminates if there are
no infinite sequences of dependency pairs \cite{arts00tcs} from first-order
rewriting to dependently-typed higher-order rewriting. Then,
following \cite{wahlstedt07phd}, we showed how to use Lee, Jones and
Ben-Amram's size-change termination criterion to prove the absence of
such infinite sequences \cite{lee01popl}.

This extends Wahlstedt's work \cite{wahlstedt07phd} from weak to
strong normalization, and from orthogonal to locally confluent rewriting
systems. This extends the first author's work \cite{blanqui05mscs}
from orthogonal to locally confluent systems, and from systems having
a decreasing argument in each recursive call to systems with
non-increasing arguments in recursive calls. Finally, this also
extends previous works on static dependency pairs
\cite{kusakari07aaecc}\hide{suzuki11pro}\hide{kop11rta} from simply-typed $\l$-calculus to
dependent types modulo rewriting.

To get this result, we assumed local confluence. However, one often
uses termination to check (local) confluence. Fortunately, there are
confluence criteria not based on termination. The most famous one is
(weak) orthogonality, that is, when the system is left-linear and has
no critical pairs (or only trivial ones) \cite{oostrom94phd}, as it is
the case in functional programming languages. A more general one is
when critical pairs are ``development-closed'' \cite{oostrom97tcs}.

This work can be extended in various directions.

First, our tool is currently limited to PFP rules,
that is, to rules where higher-order variables are direct subterms of
the left-hand side. To have higher-order variables in deeper subterms
like in Example \ref{expl-ord},
we need to define a more complex interpretation of types, following
\cite{blanqui05mscs}.

Second, to handle recursive calls in such systems, we also
need to use an ordering more complex than the subterm ordering
when computing the matrices labeling the SCT call graph.
The ordering needed for handling Example \ref{expl-ord} is the
``structural ordering'' of \tool{Coq} and \tool{Agda}
\cite{coquand92types,blanqui16tcs}. Relations other than subterm have
already been considered in SCT but in a first-order setting only
\cite{thiemann05aaecc}.

But we may want to go further because the structural ordering is not
enough to handle the following system which is not accepted by
\tool{Agda}:

\begin{example}[Division]\label{expl-div}
  $m${\tt/}$n$ computes $\lceil\frac{m}{n}\rceil$.
\begin{lstlisting}[mathescape=true]
symbol minus: Nat$\A$Nat$\A$Nat             set infix 1 "-" $\coloneqq$ minus
  rule 0 - n $\ra$ 0       rule m - 0 $\ra$ m       rule (s m) - (s n) $\ra$ m - n
symbol div: Nat$\A$Nat$\A$Nat               set infix 1 "/" $\coloneqq$ div
  rule 0 / (s n) $\ra$ 0    rule (s m) / (s n) $\ra$ s ((m - n) / (s n))
\end{lstlisting}
\end{example}

\hide{\footnote{For the sake of simplicity, we do not use the lists annotated with type and lengths of Example \ref{expl-list-poly}.}}

\hide{
\begin{example}[Map function on Rose trees]\label{expl-div}
\begin{lstlisting}[mathescape=true]
symbol tree : TYPE    symbol list_tree : TYPE
  symbol node : list_tree $\A$ tree
  symbol nil : list_tree
  symbol cons : tree $\A$ list_tree $\A$ list_tree
symbol map_list : (tree $\A$ tree) $\A$ list_tree $\A$ list_tree
  map_list f nil $\ra$ nil
  map_list f (cons x l) $\ra$ cons (f x) (map_list f l)
symbol map : (tree $\A$ tree) $\A$ tree $\A$ tree
  map f (node(cons x l)) $\ra$ node(cons (map f x) (map_list (map f) l))
\end{lstlisting}
\end{example}
}
\hide{Cet exemple est très différent de celui avec div et moins,
puisque cette fois-ci on a un appel récursif partiellement appliqué,
donc il ne suffit pas de modifier SCT,
il faut modifier même notre définition des DP pour faire passer cet exemple.
Je ne suis pas sûr qu'on veuille raconter ça ici.}

A solution to handle this system is to use arguments filterings
(remove the second argument of {\tt -}) or simple projections
\cite{hirokawa07ic}. Another one is to extend the type system with
size annotations as in \tool{Agda} and compute the SCT matrices by
comparing the size of terms instead of their structure
\cite{abel10par,blanqui18jfp}. In our example, the size of {\tt m - n}
is smaller than or equal to the size of {\tt m}. One can deduce this
by using user annotations like in \tool{Agda}, or by using heuristics
\cite{chin01hosc}.

Another interesting extension would be to handle function calls with locally
size-increasing arguments like in the following example:
\begin{lstlisting}[mathescape=true]
     rule f x $\ra$ g (s x)           rule g (s (s x)) $\ra$ f x
\end{lstlisting}
where the number of {\tt s}'s strictly decreases between two calls to
{\tt f} although the first rule makes the number of {\tt s}'s
increase. Hyvernat enriched SCT to handle such systems \cite{hyvernat14lmcs}.

\smallskip
{\bf Acknowledgments.} The authors thank the anonymous referees for their comments, which have improved the quality of this article.
%%% Local Variables:
%%%   TeX-master: "main"
%%%   mode: latex
%%%   mode: flyspell
%%%   ispell-local-dictionary: "english"
%%% End:

\bibliographystyle{plainurl}% the mandatory bibstyle
\renewcommand{\em}{}
%\bibliography{main}

\newpage
\appendix
\section{Proofs of lemmas on the interpretation}
\label{annex-interp}

\subsection{Definition of the interpretation}

\begin{lemma}\label{lem-F-mon}
 $F$ is monotone wrt inclusion.
\end{lemma}

\begin{proof}
  We first prove that $D$ is monotone. Let $I\sle J$ and $T\in D(I)$.
  We have to show that $T\in
  D(J)$. To this end, we have to prove (1) $T\in\SN$ and (2) if
  $T\a^*(x:A)B$ then $A\in\dom(J)$ and, for all $a\in J(A)$,
  $B[x\to a]\in\dom(J)$:
  \begin{enumerate}
  \item Since $T\in D(I)$, we have $T\in\SN$.
  \item Since $T\in D(I)$ and $T\a^*(x:A)B$, we have $A\in\dom(I)$
    and, for all $a\in I(A)$, $B[x\to a]\in\dom(I)$. Since $I\sle J$, we
    have $\dom(I)\sle\dom(J)$ and $J(A)=I(A)$ since $I$ and $J$ are
    functional relations. Therefore, $A\in\dom(J)$ and, for all $a\in
    I(A)$, $B[x\to a]\in\dom(J)$.
  \end{enumerate}

  We now prove that $F$ is monotone. Let $I\sle J$ and $T\in D(I)$.
  We have to show that
  $F(I)(T)=F(J)(T)$. First, $T\in D(J)$ since $D$ is
  monotone.

  If ${\nf{T}}={(x:A)B}$, then $F(I)(T)=\product{a}{I(A)}{I(B[x\to a])}$ and
  $F(J)(T)=\product{a}{J(A)}{J(B[x\to a])}$. Since $T\in D(I)$, we have
  $A\in\dom(I)$ and, for all $a\in I(A)$, $B[x\to a]\in\dom(I)$. Since
  $\dom(I)\sle\dom(J)$, we have $J(A)=I(A)$ and, for all $a\in I(A)$,
  $J(B[x\to a])=I(B[x\to a])$. Therefore, $F(I)(T)=F(J)(T)$.

  Now, if $\nf{T}$ is not a product, then
  $F(I)(T)=F(J)(T)=\SN$.
\end{proof}

%%%%%%%%%%%%%%%%%%%%%%%%%%%%%%%%%%%%%%%%%%%%%%%%%%%%%%%%%%%%%%%%%%%%%%%%%%%%%%
\subsection{Computability predicates}

\begin{lemma}\label{lem-comp-pred-dom}
  $\cD$ is a computability predicate.
\end{lemma}

\begin{proof}
Note that $\cD=D(\cI)$.
\begin{enumerate}
\item $\cD\sle\SN$ by definition of $D$.
\item Let $T\in\cD$ and $T'$ such that $T\a T'$. We have
  $T'\in\SN$ since $T\in\SN$. Assume now that $T'\a^*(x:A)B$. Then,
  $T\a^*(x:A)B$, $A\in\cD$ and, for all $a\in \cI(A)$,
  $B[x\to a]\in\cD$. Therefore, $T'\in\cD$.
\item Let $T$ be a neutral term such that $\red{T}\sle\cD$. Since
  $\cD\sle\SN$, $T\in\SN$. Assume now that $T\a^*(x:A)B$. Since
  $T$ is neutral, there is $U\in{\red{T}}$ such that
  $U\a^*(x:A)B$. Therefore, $A\in\cD$ and, for all $a\in \cI(A)$,
  $B[x\to a]\in\cD$.\qedhere
\end{enumerate}
\end{proof}

\begin{lemma}\label{lem-comp-pred-prod}
  If $P\in\bP$ and, for all $a\in P$, $Q(a)\in\bP$, then $\product{a}{P}{Q(a)}\in\bP$.
\end{lemma}

\begin{proof}
Let $R=\product{a}{P}{Q(a)}$.
\begin{enumerate}
\item Let $t\in R$. We have to prove that $t\in\SN$. Let
  $x\in\bV$. Since $P\in\bP$, $x\in P$. So, $tx\in Q(x)$. Since
  $Q(x)\in\bP$, $Q(x)\sle\SN$. Therefore, $tx\in\SN$, and $t\in\SN$.
\item Let $t\in R$ and $t'$ such that $t\a t'$. We have to prove that
  $t'\in R$. Let $a\in P$. We have to prove that $t'a\in Q(a)$. By
  definition, $ta\in Q(a)$ and $ta\a t'a$. Since $Q(a)\in\bP$, $t'a\in
  Q(a)$.
\item Let $t$ be a neutral term such that $\red{t}\sle R$. We have to
  prove that $t\in R$. Hence, we take $a\in P$ and prove that $ta\in
  Q(a)$. Since $P\in\bP$, we have $a\in\SN$ and ${\a^*\!\!(a)}\sle{P}$. We
  now prove that, for all $b\in{\a^*\!\!(a)}$, $tb\in Q(a)$, by induction on
  $\a$. Since $t$ is neutral, $tb$ is neutral too and it suffices to
  prove that $\red{tb}\sle Q(a)$. Since $t$ is neutral,
  ${\red{tb}}={\red{t}b\cup t\red{b}}$. By induction hypothesis,
  $t\red{b}\sle Q(a)$. By assumption, $\red{t}\sle R$. So,
  $\red{t}a\sle Q(a)$. Since $Q(a)\in\bP$, $\red{t}b\sle Q(a)$
  too. Therefore, $ta\in Q(a)$ and $t\in R$.\qedhere
\end{enumerate}
\end{proof}

\begin{lemma}\label{lem-comp-pred-type-int}
  For all $T\in\cD$, $\cI(T)$ is a computability predicate.
\end{lemma}

\begin{proof}
  Since $\cF_p(\bT,\bP)$ is a chain-complete poset, it suffices to
  prove that $\cF_p(\bT,\bP)$ is closed by $F$. Assume that
  $I\in\cF_p(\bT,\bP)$. We have to prove that $F(I)\in\cF_p(\bT,\bP)$,
  that is, for all $T\in D(I)$, $F(I)(T)\in\bP$. There are two
  cases:
  \begin{itemize}
  \item If ${\nf{T}}={(x:A)B}$,
    then $F(I)(T)=\product{a}{I(A)}{I(B[x\to a])}$.
    By assumption, $I(A)\in\bP$ and, for $a\in I(A)$,
    $I(B[x\to a])\in\bP$. Hence, by Lemma \ref{lem-comp-pred-prod},
    $F(I)(T)\in\bP$.
  \item Otherwise, $F(I)(T)=\SN\in\bP$.\qedhere
  \end{itemize}
\end{proof}

\noindent{\textcolor{darkgray}{$\blacktriangleright$}}
{\bf Lemma \ref{lem-props}\ref{lem-comp-pred-int}.}
{\it For all terms $T$ and substitutions $\s$, $\I{T}_\s\in\bP$.}

\begin{proof}
  By induction on $T$. If $T=s$, then $\I{T}_\s=\cD\in\bP$ by
  Lemma \ref{lem-comp-pred-dom}. If $T=(x:A)K\in\bK$, then
  $\I{T}_\s=\product{a}{\I{A}_\s}{\I{K}_{[x\to a,\s]}}$. By induction
  hypothesis, $\I{A}_\s\in\bP$ and, for all $a\in\I{A}_\s$,
  $\I{K}_{[x\to a,\s]}\in\bP$. Hence, by Lemma
  \ref{lem-comp-pred-prod}, $\I{T}_\s\in\bP$. If
  $T\notin\bK\cup\{\kind\}$ and $T\s\in\cD$, then
  $\I{T}_\s=\cI(T\s)\in\bP$ by Lemma
  \ref{lem-comp-pred-type-int}. Otherwise, $\I{T}_\s=\SN\in\bP$.
\end{proof}

%%%%%%%%%%%%%%%%%%%%%%%%%%%%%%%%%%%%%%%%%%%%%%%%%%%%%%%%%%%%%%%%%%%%%%%%%%%%%%
\subsection{Invariance by reduction}

We now prove that the interpretation is invariant by reduction.

\begin{lemma}\label{lem-type-int-red}
  If $T\in\cD$ and $T\a T'$, then $\cI(T)=\cI(T')$.
\end{lemma}

\begin{proof}
  First note that $T'\in\cD$ since $\cD\in\bP$. Hence, $\cI(T')$ is well defined. Now, we have $T\in\SN$
  since $\cD\sle\SN$. So, $T'\in\SN$ and, by local confluence and Newman's lemma,
  $\nf{T}=\nf{T'}$. If $\nf{T}=(x:A)B$
  then $\cI(T)=\product{a}{\cI(A)}{\cI(B[x\to a])}=\cI(T')$.
  Otherwise, $\cI(T)=\SN=\cI(T')$.
\end{proof}

\noindent{\textcolor{darkgray}{$\blacktriangleright$}}
{\bf Lemma \ref{lem-props}\ref{lem-int-red}.}
{\it If $T$ is typable, $T\s\in\cD$ and $T\a T'$, then $\I{T}_\s=\I{T'}_\s$.}

\begin{proof}
  By assumption, there are $\G$ and $U$ such that $\G\th T:U$. Since
  $\a$ preserves typing, we also have $\G\th T':U$. So, $T\neq\kind$, and
  $T'\neq\kind$. Moreover, $T\in\bK$ iff $T'\in\bK$ since $\G\th T:\kind$ iff
  $T\in\bK$ and $T$ is typable. In addition, we have $T'\s\in\cD$
  since $T\s\in\cD$ and $\cD\in\bP$.

  We now prove the result, with $T\a^=T'$ instead of $T\a T'$, by
  induction on $T$. If $T\notin\bK$, then $T'\notin\bK$ and, since
  $T\s,T'\s\in\cD$, $\I{T}_\s=\cI(T\s)=\cI(T'\s)=\I{T'}_\s$ by Lemma
  \ref{lem-type-int-red}. If $T=\type$, then
  $\I{T}_\s=\cD=\I{T'}_\s$. Otherwise, $T=(x:A)K$ and
  $T'=(x:A')K'$ with $A\a^=A'$ and $K\a^=K'$. By inversion, we have
  $\G\th A:\type$, $\G\th A':\type$, $\G,x:A\th K:\kind$ and $\G,x:A'\th
  K':\kind$. So, by induction hypothesis, $\I{A}_\s=\I{A'}_\s$ and, for
  all $a\in\I{A}_\s$, $\I{K}_{\s'}=\I{K'}_{\s'}$, where $\s'=[x\to
    a,\s]$. Therefore, $\I{T}_\s=\I{T'}_\s$.
\end{proof}

\noindent{\textcolor{darkgray}{$\blacktriangleright$}}
{\bf Lemma \ref{lem-props}\ref{lem-int-red-subs}.}
{\it If $T$ is typable, $T\s\in\cD$ and $\s\a\s'$, then $\I{T}_\s=\I{T}_{\s'}$.}

\begin{proof}
  By induction on $T$.
  \begin{itemize}
  \item If $T\in\mb{S}$, then $\I{T}_\s=\cD=\I{T}_{\s'}$.
  \item If $T=(x:A)K$ and $K\in\bK$,
    then $\I{T}_\s=\product{a}{\I{A}_\s}{\I{K}_{[x\to a,\s]}}$
    and $\I{T}_{\s'}=\product{a}{\I{A}_{\s'}}{\I{K}_{[x\to a,\s']}}$.
    By induction hypothesis,
    $\I{A}_\s=\I{A}_{\s'}$ and, for all $a\in\I{A}_\s$, $\I{K}_{[x\to
        a,\s]}=\I{K}_{[x\to a,\s']}$. Therefore,
    $\I{T}_\s=\I{T}_{\s'}$.
  \item If $T\s\in\cD$, then $\I{T}_\s=\cI(T\s)$ and
    $\I{T}_{\s'}=\cI(T\s')$. Since $T\s\a^*T\s'$, by Lemma
    \ref{lem-props}\ref{lem-int-red}, $\cI(T\s)=\cI(T\s')$.
  \item Otherwise, $\I{T}_\s=\SN=\I{T}_{\s'}$.\qedhere
  \end{itemize}
\end{proof}

%%%%%%%%%%%%%%%%%%%%%%%%%%%%%%%%%%%%%%%%%%%%%%%%%%%%%%%%%%%%%%%%%%%%%%%%%%%%%%
\subsection{Adequacy of the interpretation}

\noindent{\textcolor{darkgray}{$\blacktriangleright$}}
{\bf Lemma \ref{lem-props}\ref{lem-int-prod}.}
{\it If $(x:A)B$ is typable, $((x:A)B)\s\in\cD$ and
  $x\notin\dom(\s)\cup\FV(\s)$,
  then $\I{(x:A)B}_\s=\product{a}{\I{A}_\s}{\I{B}_{[x\to a,\s]}}$.}

\begin{proof}
  If $B$ is a kind, this is immediate. Otherwise, since
  $((x:A)B)\s\in\cD$, $\I{(x:A)B}_\s=\cI(((x:A)B)\s)$. Since
  $x\notin\dom(\s)\cup\FV(\s)$, we have $((x:A)B)\s=(x:A\s)B\s$. Since
  $(x:A\s)B\s\in\cD$ and $\cD\sle\SN$, we have
  $\I{(x:A)B}_\s=\product{a}{\cI(\nf{A\s})}{\cI((\nf{B\s})[x\to a])}$.

  Since $(x:A)B$ is typable, $A$ is of type $\type$ and
  $A\notin\bK\cup\{\kind\}$. Hence, $\I{A}_\s=\cI(A\s)$ and, by Lemma
  \ref{lem-type-int-red}, $\cI(A\s)=\cI(\nf{A\s})$.

  Since $(x:A)B$ is typable and not a kind, $B$ is of type $\type$ and
  $B\notin\bK\cup\{\kind\}$. Hence, $\I{B}_{[x\to a,\s]}=\cI(B[x\to
    a,\s])$. Since $x\notin\dom(\s)\cup\FV(\s)$, $B[x\to
    a,\s]=(B\s)[x\to a]$. Hence, $\I{B}_{[x\to a,\s]}=\cI((B\s)[x\to a])$ and,
  by Lemma \ref{lem-type-int-red},
  $\cI((B\s)[x\to a])=\cI((\nf{B\s})[x\to a])$.

  Therefore, $\I{(x:A)B}_\s=\product{a}{\I{A}_\s}{\I{B}_{[x\to a,\s]}}$.
\end{proof}

Note that, by iterating this lemma, we get that $v\in\I{\prod\vx\vT U}$
iff, for all $\vt$ such that $[\vx\to\vt]\models\vx:\vT$,
$v\vt\in\I{U}_{[\vx\to\vt]}$.

\bigskip
\noindent{\textcolor{darkgray}{$\blacktriangleright$}}
{\bf Lemma \ref{lem-props}\ref{lem-int-subs}.}
{\it If $\D\th U:s$, $\G\th\g:\D$ and $U\g\s\in\cD$, then
  $\I{U\g}_\s=\I{U}_{\g\s}$.}

\begin{proof}
  We proceed by induction on $U$. Since $\D\th U:s$ and $\G\th\g:\D$,
  we have $\G\th U\g:s$.
  \begin{itemize}
   \item If $s=\type$, then $U,U\g\notin\bK\cup\{\kind\}$
  and $\I{U\g}_\s=\cI(U\g\s)=\I{U}_{\g\s}$ since
  $U\g\s\in\cD$.
  \item Otherwise, $s=\kind$ and $U\in\bK$.
   \begin{itemize}
    \item If $U=\type$,
     then $\I{U\g}_\s=\cD=\I{U}_{\g\s}$.
    \item Otherwise, $U=(x:A)K$ and,
      by Lemma \ref{lem-props}\ref{lem-int-prod},
       $\I{U\g}_\s=\product{a}{\I{A\g}_\s}{\I{K\g}_{[x\to a,\s]}}$
       and $\I{U}_{\g\s}=\product{a}{\I{A}_{\g\s}}{\I{K}_{[x\to a,\g\s]}}$.
       By induction hypothesis,
       $\I{A\g}_\s=\I{A}_{\g\s}$ and, for all $a\in\I{A\g}_\s$,
       $\I{K\g}_{[x\to a,\s]}=\I{K}_{\g[x\to a,\s]}$. Wlog we can assume $x\notin\dom(\g)\cup\FV(\g)$. So, $\I{K}_{\g[x\to
       a,\s]}=\I{K}_{[x\to a,\g\s]}$.\qedhere
   \end{itemize}
  \end{itemize}
\end{proof}

\bigskip
\noindent{\textcolor{darkgray}{$\blacktriangleright$}}
{\bf Lemma \ref{lem-props}\ref{lem-comp-abs}.}
{\it Let $P$ be a computability predicate and $Q$ a $P$-indexed family of
  computability predicates such that $Q(a')\sle Q(a)$ whenever $a\a
  a'$. Then, $\l x:A.b\in\product{a}{P}{Q(a)}$ whenever $A\in\SN$ and, for
  all $a\in P$, $b[x\to a]\in Q(a)$.}

\begin{proof}
  Let $a_0\in P$. Since $P\in\bP$, we have $a_0\in\SN$ and $x\in P$. Since
  $Q(x)\in\bP$ and $b=b[x\to x]\in Q(x)$, we have $b\in\SN$.
  Let $a\in{\a^*(a_0)}$. We can prove
  that $(\abs xAb)a\in Q(a_0)$ by induction on $(A,b,a)$ ordered by
  $(\a,\a,\a)_{\mr{prod}}$. Since $Q(a_0)\in\bP$ and $(\abs xAb)a$ is
  neutral, it suffices to prove that $\red{(\abs xAb)a}\sle
  Q(a_0)$. If the reduction takes place in $A$, $b$ or $a$, we can
  conclude by induction hypothesis. Otherwise, $(\abs xAb)a\a b[x\to a]\in
  Q(a)$ by assumption. Since $a_0\a^*a$ and $Q(a')\sle Q(a)$ whenever
  $a\a a'$, we have $b[x\to a]\in Q(a_0)$.
\end{proof}

%%%%%%%%%%%%%%%%%%%%%%%%%%%%%%%%%%%%%%%%%%%%%%%%%%%%%%%%%%%%%%%%%%%%
\section{Termination proof of Example \ref{expl-list-poly}}
\label{annex-matrices}

Here is the comprehensive list of dependency pairs in the example:
\begin{lstlisting}[escapeinside={*}{*}]
*\color{red} A:*               El (arrow a b) > El a
*\color{red} B:*               El (arrow a b) > El b
*\color{red} C:*                    (s p) + q > p + q
*\color{red} D:*   app a _ (cons _ x p l) q m > p + q
*\color{red} E:*   app a _ (cons _ x p l) q m > app a p l q m
*\color{red} F:*len_fil a f _ (cons _ x p l)  > len_fil_aux (f x) a f p l
*\color{red} G:*len_fil a f _ (app _ p l q m) >
                        (len_fil a f p l) + (len_fil a f q m)
*\color{red} H:*len_fil a f _ (app _ p l q m) > len_fil a f p l
*\color{red} I:*len_fil a f _ (app _ p l q m) > len_fil a f q m
*\color{red} J:*    len_fil_aux true  a f p l > len_fil a f p l
*\color{red} K:*    len_fil_aux false a f p l > len_fil a f p l
*\color{red} L:*    fil a f _ (cons _ x p l)  > fil_aux (f x) a f x p l
*\color{red} M:*    fil a f _ (app _ p l q m) >
                        app a (len_fil a f p l) (fil a f p l)
                              (len_fil a f q m) (fil a f q m)
*\color{red} N:*    fil a f _ (app _ p l q m) > len_fil a f p l
*\color{red} O:*    fil a f _ (app _ p l q m) > fil a f p l
*\color{red} P:*    fil a f _ (app _ p l q m) > len_fil a f q m
*\color{red} Q:*    fil a f _ (app _ p l q m) > fil a f q m
*\color{red} R:*      fil_aux true  a f x p l > len_fil a f p l
*\color{red} S:*      fil_aux true  a f x p l > fil a f p l
*\color{red} T:*      fil_aux false a f x p l > fil a f p l
\end{lstlisting}

The whole callgraph is depicted below.
The letter associated to each matrix corresponds to the
dependency pair presented above and in example \ref{expl-mult}, except for TC 's which comes from the computation
of the transitive closure and labels dotted edges.

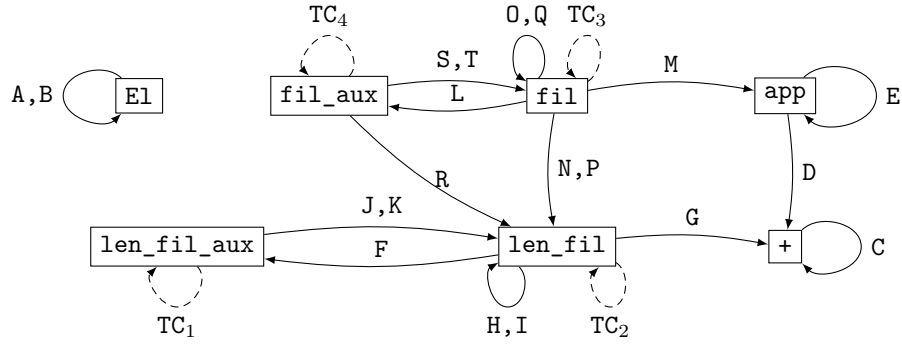
\begin{center}\tt
  \begin{tikzpicture}
    \node[draw] (filter) at (2,0) {fil};
    \node[draw] (f_aux) at (-1,0) {fil\_aux};
    \node[draw] (l_filter) at (2,-2) {len\_fil};
    \node[draw] (l_f_aux) at (-3,-2) {len\_fil\_aux};
    \node[draw] (el) at (-3.5,0) {El};
    \node[draw] (app) at (5,0) {app};
    \node[draw] (pl) at (5,-2) {+};
    \draw[>=latex,->] (el) to[out=135,in=90] (-4.5,0) to[out=-90,in=-135] (el);
    \node[left] (ell) at (-4.5,0) {A,B};
    \draw[>=latex,->] (pl) to[out=45,in=90] (6,-2) to[out=-90,in=-45] (pl);
    \node[right] (pll) at (6,-2) {C};
    \draw[>=latex,->] (app) to[bend left=8] node[midway,right] {D} (pl);
    \draw[>=latex,->] (app) to[out=45,in=90] (6.2,0) to[out=-90,in=-45] (app);
    \node[right] (appl) at (6.2,0) {E};
    \draw[>=latex,->] (l_filter) to[bend left=8] node[midway,above] {F} (l_f_aux);
    \draw[>=latex,->] (l_filter) to[bend left=8] node[midway,above] {G} (pl);
    \draw[>=latex,->] (1.5,-2.27) to[out=-45,in=0] (1.35,-2.8) to[out=180,in=-150] (1.23,-2.21);
    \node[below] (l_filterl) at (1.35,-2.8) {H,I};
    \draw[>=latex,->] (l_f_aux) to[bend left=8] node[midway,above] {J,K} (l_filter);
    \draw[>=latex,->] (filter) to[bend left=10] node[midway,above] {L} (f_aux);
    \draw[>=latex,->] (f_aux) to[bend right=10] node[midway,above,right] {R} (l_filter);
    \draw[>=latex,->] (f_aux) to[bend left=10] node[midway,above] {S,T} (filter);
    \draw[>=latex,->] (filter) to[bend left=10] node[midway,above] {M} (app);
    \draw[>=latex,->] (filter) to[bend right=10] node[midway,right] {N,P} (l_filter);
    \draw[>=latex,->] (1.75,0.23) to[out=45,in=0] (1.6,0.8) to[out=180,in=150] (1.6,0.19);
    \node[above] (filterl) at (1.6,0.8) {O,Q};
    \draw[>=latex,->,densely dashed] (f_aux) to[out=45,in=0] (-1,0.8) to[out=180,in=135] (f_aux);
    \node[above] (f_auxTC) at (-1,0.8) {TC${}_4$};
    \draw[>=latex,->,densely dashed] (2.4,0.19) to[out=30,in=0] (2.4,0.8) to[out=180,in=135] (2.25,0.23);
    \node[above] (filterTC) at (2.4,0.8) {TC${}_3$};
    \draw[>=latex,->,densely dashed] (l_f_aux) to[out=-45,in=0] (-3,-2.8) to[out=180,in=-135] (l_f_aux);
    \node[below] (l_f_auxTC) at (-3,-2.8) {TC${}_1$};
    \draw[>=latex,->,densely dashed] (2.77,-2.21) to[out=-30,in=0] (2.7,-2.8) to[out=180,in=-135] (2.55,-2.27);
    \node[below] (l_filterTC) at (2.7,-2.8) {TC${}_2$};
    % \draw[>=latex,->] (app) to[bend left=10] node[midway,right] {$\beta$} (el);
    % \draw[>=latex,->] (l_filter) to[bend left=10] node[midway,right] {$\delta$,$\epsilon$} (el);
    % \draw[>=latex,->] (l_f_aux) to[bend left=10] node[midway,right] {$\zeta$,$\eta$} (el);
    % \draw[>=latex,->] (filter) to[bend left=10] node[midway,right] {$\theta$,$\kappa$,$\nu$} (el);
    % \draw[>=latex,->] (filter) to node[pos=0.6,above] {$\mu$} (l_f_aux);
    % \draw[>=latex,->] (filter) to[bend left=10] node[midway,right] {$o$} (pl);
    % \draw[>=latex,->] (f_aux) to[bend left=10] node[midway,right] {$\pi$,$\tau$} (el);
    % \draw[>=latex,->] (f_aux) to node[midway,left] {$\sigma$,$\phi$} (l_f_aux);
  \end{tikzpicture}
\end{center}

%%%%%%%%%%%%%%%%%%%%%%%%%%%%%%%%%%%%%%%%%%%%%%%%%%%%%%%%%%%%%%%%%%%%%%

The argument {\tt a} is omitted everywhere on the matrices presented below:

\begin{center}\small
{\tt A,B}=$\left(\begin{smallmatrix}
          -1\\
         \end{smallmatrix}\right)$,
{\tt C}=$\left(\begin{smallmatrix}
          -1     & \infty \\
          \infty & 0      \\
         \end{smallmatrix}\right)$,
{\tt D}=$\left(\begin{smallmatrix}
          \infty & \infty\\
          -1     & \infty\\
          \infty & 0     \\
          \infty & \infty\\
         \end{smallmatrix}\right)$,
{\tt E}=$\left(\begin{smallmatrix}
          \infty & \infty & \infty & \infty\\
          -1     & -1     & \infty & \infty\\
          \infty & \infty & 0      & \infty\\
          \infty & \infty & \infty & 0     \\
         \end{smallmatrix}\right)$,
{\tt F}=$\left(\begin{smallmatrix}
          \infty & 0      & \infty & \infty\\
          \infty & \infty & \infty & \infty\\
          \infty & \infty & -1     & -1    \\
         \end{smallmatrix}\right)$,
{\tt J}={\tt K}=$\left(\begin{smallmatrix}
          \infty & \infty & \infty\\
          0      & \infty & \infty\\
          \infty & 0      & \infty\\
          \infty & \infty & 0     \\
         \end{smallmatrix}\right)$,
{\tt G}=$\left(\begin{smallmatrix}
          \infty & \infty\\
          \infty & \infty\\
          \infty & \infty\\
         \end{smallmatrix}\right)$,
{\tt H}={\tt I}={\tt N}={\tt O}={\tt P}={\tt Q}=$\left(\begin{smallmatrix}
          0      & \infty & \infty\\
          \infty & \infty & \infty\\
          \infty & -1     & -1    \\
         \end{smallmatrix}\right)$,
{\tt L}=$\left(\begin{smallmatrix}
          \infty & 0      & \infty & \infty & \infty\\
          \infty & \infty & \infty & \infty & \infty\\
          \infty & \infty & -1     & -1     & -1    \\
         \end{smallmatrix}\right)$,
{\tt M}=$\left(\begin{smallmatrix}
          \infty & \infty & \infty & \infty\\
          \infty & \infty & \infty & \infty\\
          \infty & \infty & \infty & \infty\\
         \end{smallmatrix}\right)$,
{\tt R}={\tt S}={\tt T}=$\left(\begin{smallmatrix}
          \infty & \infty & \infty\\
          0      & \infty & \infty\\
          \infty & \infty & \infty\\
          \infty & 0      & \infty\\
          \infty & \infty & 0     \\
         \end{smallmatrix}\right)$.
\end{center}
Which leads to the matrices labeling a loop in the transitive closure:
\begin{center}
{\tt TC}${}_1$={\tt J}$\times${\tt F}=$\left(\begin{smallmatrix}
                                                    \infty & \infty & \infty & \infty\\
                                                    \infty & 0      & \infty & \infty\\
                                                    \infty & \infty & \infty & \infty\\
                                                    \infty & \infty & -1     & -1    \\
                                                   \end{smallmatrix}\right)$,
{\tt TC}${}_4$={\tt S}$\times${\tt L}=$\left(\begin{smallmatrix}
                                                    \infty & \infty & \infty & \infty & \infty\\
                                                    \infty & 0      & \infty & \infty & \infty\\
                                                    \infty & \infty & \infty & \infty & \infty\\
                                                    \infty & \infty & \infty & \infty & \infty\\
                                                    \infty & \infty & -1     & -1     & -1    \\
                                                   \end{smallmatrix}\right)$,
{\tt TC}${}_3$={\tt L}$\times${\tt S}={\tt TC}${}_2$={\tt F}$\times${\tt J}=$\left(\begin{smallmatrix}
                                                    0      & \infty & \infty\\
                                                    \infty & \infty & \infty\\
                                                    \infty & -1     & -1    \\
                                                   \end{smallmatrix}\right)$={\tt O}={\tt H}.
\end{center}
It would be useless to compute matrices labeling edges which are not in a strongly connected component of the call-graph (like {\tt S$\times$R}),
but it is necessary to compute all the products which could label a loop,
especially to verify that all loop-labeling matrices are idempotent,
which is indeed the case here.

\hide{
As an example, we can detail the calculus performed to compute {\tt TC${}_3^2$}:
\[\left(\begin{smallmatrix}
  \min(0+0,\infty+\infty,\infty+\infty) &
    \min(0+\infty,\infty+\infty,\infty+-1) &
      \min(0+\infty,\infty+\infty,\infty+-1)\\
  \min(\infty+0,\infty+\infty,\infty+\infty) &
    \min(\infty+\infty,\infty+\infty,\infty+-1) &
      \min(\infty+\infty,\infty+\infty,\infty+-1)\\
  \min(\infty+0,-1+\infty,-1+\infty) &
    \min(\infty+\infty,-1+\infty,-1+-1) &
      \min(\infty+\infty,-1+\infty,-1+-1)\\
 \end{smallmatrix}\right)\]
}

We now check that this system is well-structured. For each rule
$f\vl\a r$, we take the environment $\D_{f\vl\a r}$ made of all
the variables of $r$ with the following types:
{\tt a:Set, b:Set, p:$\bN$, q:$\bN$, x:El a, l:$\bL$ a p, m:$\bL$ a q, f:El a$\A\bB$}.

The precedence infered for this example is the smallest containing:
\begin{itemize}
\item comparisons linked to the typing of symbols:
{\begin{center}\tt
\begin{tabular}{rlrl}
Set &$\preceq$ arrow &
Set,$\bL$,0 &$\preceq$ nil\\
Set &$\preceq$ El&
Set,El,$\bN$,$\bL$,s &$\preceq$ cons\\
$\bB$ &$\preceq$ true &
Set,$\bN$,$\bL$,+ &$\preceq$ app\\
$\bB$ &$\preceq$ false &
Set,El,$\bB$,$\bN$,$\bL$ &$\preceq$ len\_fil\\
$\bN$ &$\preceq$ 0 &
$\bB$,Set,El,$\bN$,$\bL$ &$\preceq$ len\_fil\_aux\\
$\bN$ &$\preceq$ s&
Set,El,$\bB$,$\bN$,$\bL$,len\_fil &$\preceq$ fil\\
$\bN$ &$\preceq$ + &
$\bB$,Set,El,$\bN$,$\bL$,len\_fil\_aux &$\preceq$ fil\_aux \\
Set,$\bN$ &$\preceq$ $\bL$&\\
\end{tabular}
\end{center}}
\item and comparisons related to calls:
{\begin{center}\tt
\begin{tabular}{rlrl}
s &$\preceq$ + &
s,len\_fil &$\preceq$ len\_fil\_aux \\
cons,+ &$\preceq$ app &
nil,fil\_aux,app,len\_fil &$\preceq$ fil \\
0,len\_fil\_aux,+ &$\preceq$ len\_fil &
fil,cons,len\_fil &$\preceq$ fil\_aux\\
\end{tabular}
\end{center}}
\end{itemize}

This precedence can be sum up in the following diagram, where symbols in the same box are equivalent:
\begin{center}\tt
  \begin{tikzpicture}
    \node[draw] (filter) at (2,6){fil,fil\_aux};
    \node[draw] (l_filter) at (3,4.5) {len\_fil,len\_fil\_aux};
    \node[draw] (append) at (-1,4.5) {app};
    \node[draw] (true) at (4,1.5) {true};
    \node[draw] (false) at (6,1.5) {false};
    \node[draw] (B) at (5,0) {$\bB$};
    \node[draw] (cons) at (-2,3) {cons};
    \node[draw] (nil) at (0,3) {nil};
    \node[draw] (plus) at (3,3) {+};
    \node[draw] (L) at (-1,1.5) {$\bL$};
    \node[draw] (arrow) at (-4,1.5) {arrow};
    \node[draw] (El) at (-3,1.5) {El};
    \node[draw] (zero) at (1,1.5) {0};
    \node[draw] (succ) at (3,1.5) {s};
    \node[draw] (Set) at (-3.5,0) {Set};
    \node[draw] (N) at (2,0) {$\bN$};
    \draw[>=latex,->] (Set) to (arrow);
    \draw[>=latex,->] (Set) to (El);
    \draw[>=latex,->] (Set) to (L);
    \draw[>=latex,->] (N) to (L);
    \draw[>=latex,->] (N) to (zero);
    \draw[>=latex,->] (N) to (succ);
    \draw[>=latex,->] (B) to (true);
    \draw[>=latex,->] (B) to (false);
    \draw[>=latex,->] (El) to (cons);
    \draw[>=latex,->] (L) to (cons);
    \draw[>=latex,->] (L) to (nil);
    \draw[>=latex,->] (zero) to (nil);
    \draw[>=latex,->] (succ) to (plus);
    \draw[>=latex,->] (cons) to (append);
    \draw[>=latex,->] (plus) to (append);
    \draw[>=latex,->] (plus) to (l_filter);
    \draw[>=latex,->] (append) to (filter);
    \draw[>=latex,->] (l_filter) to (filter);
    \draw[>=latex,->] (zero) to (l_filter);
    \draw[>=latex,->] (L) to[bend right=5] (l_filter);
    \draw[>=latex,->] (B) to[bend right=35] (l_filter);
    \draw[>=latex,->] (nil) to[bend left=15] (filter);
    \draw[>=latex,->] (succ) to (cons);
    \draw[>=latex,->] (El) to[out=15,in=-175] (l_filter);
  \end{tikzpicture}
\end{center}

%%% Local Variables:
%%%   TeX-master: "main"
%%%   mode: latex
%%%   mode: flyspell
%%%   ispell-local-dictionary: "english"
%%% End:

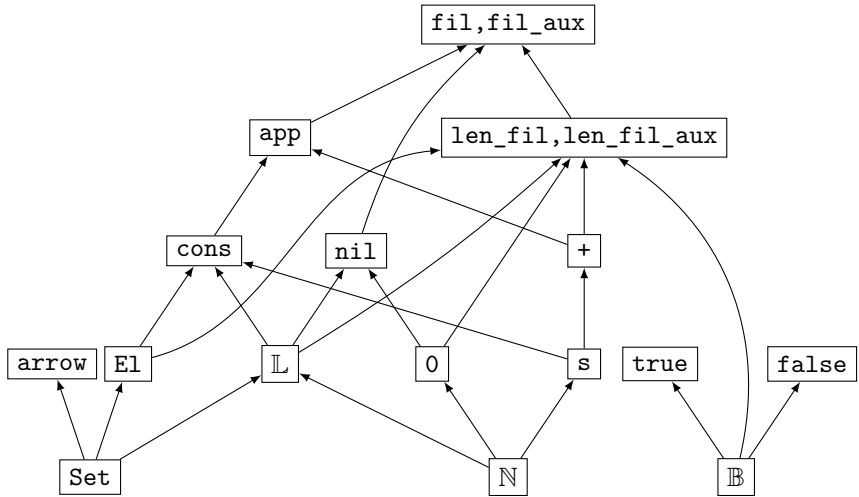
\end{document}